\documentclass[11pt,a4paper]{article}

\usepackage{amssymb}
\usepackage{amsmath, amsthm}
\usepackage{epsfig}

\usepackage{geometry}

\def\RR{{\bf R}}
\def\ZZ{{\bf Z}}

\def\mGamma{\mathit{\Gamma}}
\def\mDelta{\mathit{\Delta}}
\def\mSigma{\mathit{\Sigma}}

\numberwithin{equation}{section}
\newtheorem{Thm}{Theorem}[section]
\newtheorem{Prop}[Thm]{Proposition}
\newtheorem{Lem}[Thm]{Lemma}

\theoremstyle{definition}

\newcommand{\opt}{\mathop{\rm opt} }

\title{A dual descent algorithm   for \\ 
	node-capacitated multiflow problems 
	and its applications}
\author{Hiroshi HIRAI \\
Department of Mathematical Informatics, \\
Graduate School of Information Science and Technology,   \\
University of Tokyo, Tokyo, 113-8656, Japan.\\
\texttt{\normalsize hirai@mist.i.u-tokyo.ac.jp}}

\begin{document}

\maketitle

\begin{abstract}
	In this paper,  we develop an $O((m \log k) {\rm MSF} (n,m,1))$-time algorithm to find 
	a half-integral node-capacitated multiflow of the maximum total flow-value 
	in a network with $n$ nodes, $m$ edges, and $k$ terminals,  
	where  ${\rm MSF} (n',m',\gamma)$ denotes 
	the time complexity of solving the maximum submodular flow problem  
	in a network with $n'$ nodes, $m'$ edges, and the complexity $\gamma$ of 
	computing the exchange capacity of the submodular function describing the problem.
	By using Fujishige-Zhang algorithm for submodular flow, 
	we can find a maximum half-integral multiflow in $O(m n^3 \log k)$ time.
	This is the first combinatorial strongly polynomial time algorithm for this problem.
	Our algorithm is built on  a developing theory of discrete convex functions on certain graph structures.
	Applications include ``ellipsoid-free" combinatorial implementations 
	of a 2-approximation algorithm for the minimum node-multiway cut problem by Garg, Vazirani, and Yannakakis.
\end{abstract}


\section{Introduction}

A node-capacitated undirected network is 
a quadruple $N = (V,E,S, c)$ of node set $V$, (undirected) edge set $E$,
a specified subset $S$ of nodes, called {\em terminals}, and 
a nonnegative integer-valued node capacity $c: V \setminus S \to \ZZ_+$ on nonterminal nodes. 
An {\em $S$-path} is a path connecting distinct terminals.
A (node-capacitated) {\em multiflow} is a pair 
$({\cal P},\lambda)$ of a set ${\cal P}$ of $S$-paths and a flow-value function $\lambda: {\cal P} \to \RR_+$
satisfying the node-capacity constraint:
\begin{equation}\label{eqn:node_cap}
	\sum_{P \in {\cal P}:\  i \in V(P)} \lambda(P) \leq c(i) \quad (i \in V \setminus S). 
\end{equation}
The total flow-value of a multiflow $f = ({\cal P}, \lambda)$ 
is defined as $\sum_{P \in {\cal P}}\lambda(P)$.
A multiflow is called {\em maximum} 
if it has the maximum total flow-value among all possible multiflows.
A multiflow $f = ({\cal P},\lambda)$
is said to be {\em integral} if $\lambda$ is integer-valued, and {\em half-integral} if $2 \lambda$ is integer-valued.

In this paper, we address the problem of finding a maximum multiflow in a node-capacitated network.
This multiflow problem appeared in the work by Garg, Vazirani, and Yannakakis~\cite{GVY04} on
an approximation algorithm for node-multiway cut. 
In fact, the LP-dual of our multiflow problem is a natural LP-relaxation of 
the {\em minimum node-multiway cut problem}; see also \cite[Section 19.3]{Vazirani}.
They showed that this LP-dual always has a half-integral optimal solution.
The half-integrality of the primal problem, i.e.,
the existence of a half-integral maximum multiflow, 
was later shown by Pap~\cite{Pap07STOC, Pap08EGRES}.
He also showed that a half-integral maximum multiflow can be found 
in strongly polynomial time.
%

In these works, 
the polynomial time solvability
depends on the use of the ellipsoid method.
Thus it is natural to seek a combinatorial polynomial time algorithm.
For the case of unit node-capacity ($c(i) = 1$ for all $i \in V \setminus S$), 
Babenko~\cite{Babenko10} developed a combinatorial 
$O(mn^2)$ time algorithm to find a half-integral maximum multiflow, 
where $n$ is the number of nodes and $m$ is the number of edges; see 
Babenko and Artamonov~\cite{BabenkoArtamonov17} 
for a further improvement.
For general node-capacity,
Babenko and Karzanov~\cite{BK08ESA} developed 
a combinatorial weakly polynomial time algorithm to find 
a half-integral maximum multiflow.
Their algorithm runs in $O({\rm MF}(n,m,C) n^2 \log^2 n \log C)$ time,
where ${\rm MF}(n,m,C)$ is the time complexity of solving 
the max-flow problem in a network with $n$ nodes, $m$ edges, and 
the maximum edge-capacity $C$.

The main result of this paper is the first combinatorial
{\em strongly} polynomial time algorithm to solve the maximum node-capacitated multiflow problem.
Our algorithm uses, as a subroutine, 
an algorithm of solving the {\em maximum submodular flow problem}; 
see \cite[Section 5.5 (c)]{FujiBook}.
Let ${\rm MSF}(n,m,\gamma)$ denote the time complexity of 
solving the maximum submodular flow problem on a network with $n$ nodes, 
$m$ edges, and the time complexity $\gamma$ of computing the exchange capacity of the submodular function describing the problem.
\begin{Thm}\label{thm:main}
	There exists an $O((m \log k) {\rm MSF}(n, m, 1))$-time algorithm to find a half-integral maximum multiflow and a half-integral optimal dual solution in a network
	of $n$ nodes, $m$ edges, and $k$ terminals.
\end{Thm}
The current fastest maximum submodular flow algorithm is the push-relabel algorithm
due to Fujishige and Zhang~\cite{FZ92} of the time complexity $O(n^3 \gamma)$; see the survey \cite{FI00survey} on 
submodular flow algorithms.
Thus we can solve the problem in $O(mn^3 \log k)$ time.

%
\paragraph{Application 1: Node-multiway cut.}
A {\em node-multiway cut} is a subset $X \subseteq V \setminus S$ 
of nonterminal nodes such that the deletion of $X$ 
makes every pair of distinct terminals unreachable, 
or equivalently, $X$ meets every $S$-path.
The {\em capacity} of a node-multiway cut $X$ is defined as $\sum_{i \in X} c(i)$.
The {\em minimum node-multiway cut problem} asks to find a node-multiway cut 
with the minimum capacity.
This well-known NP-hard problem is 
naturally formulated as the following $\{ 0,1\}$-integer program:
\begin{eqnarray}
\mbox{Minimize} && \sum_{i \in V \setminus S} c(i) w(i)  \nonumber \\
\mbox{subject to } && w: V \setminus S \to \{0,1\}, \nonumber \\
  && \sum_{i \in V(P) \setminus S} w(i) \geq 1 \quad (\mbox{every $S$-path $P$}). \label{eqn:LP-dual}
\end{eqnarray}
The natural LP-relaxation obtained by relaxing
$w: V \setminus S \to \{0,1\}$ into $w: V \setminus S \to \RR_{+}$
is nothing but the LP-dual of our multiflow problem.
As mentioned above,
Garg, Vazirani, and Yannakakis~\cite{GVY04} proved 
that a half-integral optimal LP solution $w^*: V \setminus S \to \{0,1/2,1\}$ always exists, 
and is obtained from any optimal LP solution 
by a simple rounding procedure; see \cite[Section 19.3]{Vazirani}.
Then the set of nodes $i$ with $w^*(i) \geq 1/2$ is 
a 2-approximation solution of the minimum node-multiway cut problem.
This rounding algorithm needs an optimal LP solution, which is now obtained by our algorithm.
To the best of our knowledge, this is the first combinatorial strongly polynomial time 
implementation of the $2$-approximation algorithm.

Recently, 
Chakuri and Madan~\cite{ChekuriMadanSODA16} devised a simple method
to round any feasible LP solution
into a multiway cut of capacity within factor $2$.
Combining this rounding method with a fast FPTAS for multiflow~(e.g., \cite{GargKonemann07}), 
they obtain a considerably faster $(2+ \epsilon)$-approximation algorithm 
(with running time dependent on $1/\epsilon$).

\paragraph{Application 2: Integral multiflow.}
Our algorithm is also useful in the problem of finding a maximum integral multiflow.
This problem is a capacitated version of 
openly-disjoint $S$-paths packing problem 
considered by Mader~\cite{Mader7879} 
(that corresponds to the case of $c(i) = 1$ $(i \in V \setminus S)$).
Pap~\cite{Pap07STOC,Pap08EGRES} established 
the strongly polynomial time solvability of the maximum integral multiflow problem.
The first step of his algorithm is to 
find a maximum half-integral multiflow.
The second step is to construct and solve
an instance of the openly-disjoint $S$-path packing problem 
(on the graph with size polynomial in the numbers of edges and nodes in the original network).
Finally, combining the integer part of the half-integral multiflow with a solution of the packing problem,
one obtains a maximum integer multiflow.
The second step can be done by
several combinatorial polynomial time algorithms, 
including~\cite{CCG08} and \cite[Section 73.1a]{SchrijverBook}.
Our algorithm can be used in the first step, and 
makes the whole algorithm fully combinatorial.

\paragraph{Outline.}
Let us outline our algorithm and the ideas behind it, as well as the structure of the paper.
Figure~\ref{fig:outline} illustrates the outline which our argument follows. 
\begin{figure} 
	\begin{center} 
		\includegraphics[scale=0.80]{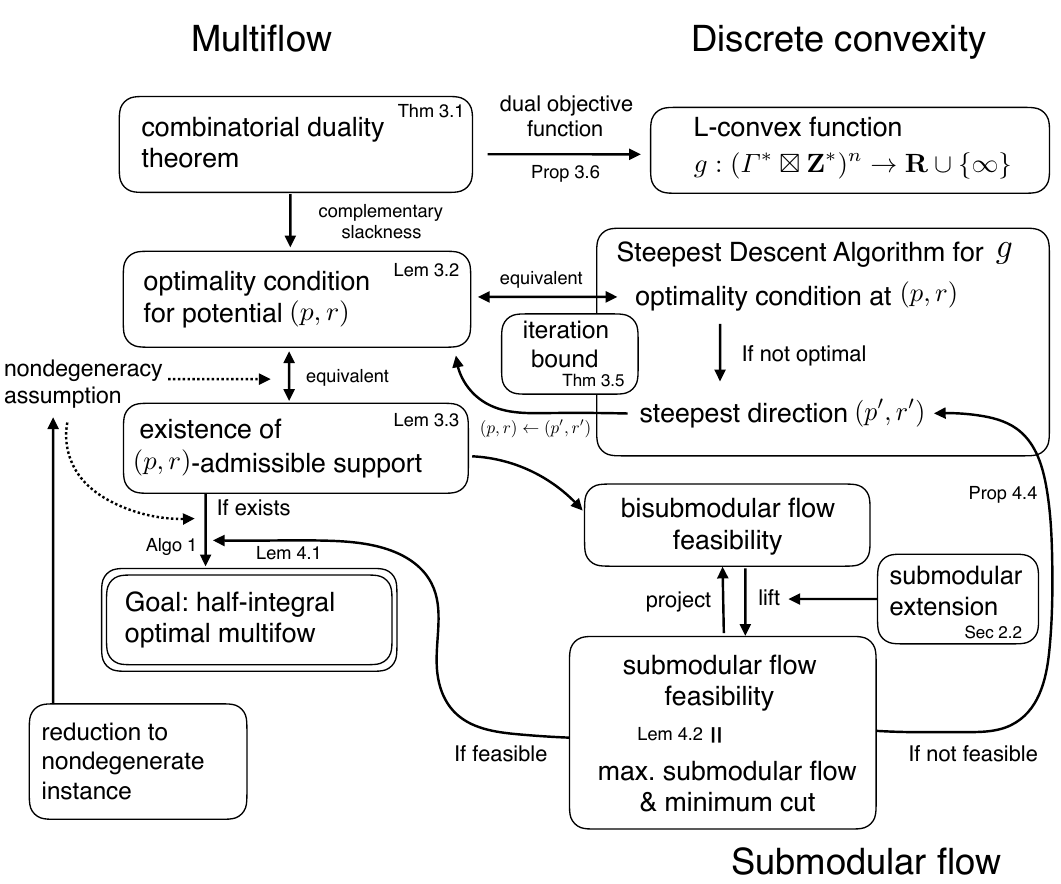}
		\caption{Outline}  
		\label{fig:outline}         
	\end{center}
\end{figure} 
Our algorithm is designed on the basis of the following  two ingredients.
One is a combinatorial duality theory for a class of node-capacitated multiflow problems~\cite{HHMPA}.
The other is a developing theory of discrete convex functions 
on certain graph structures~\cite{HH150ext,HH14extendable,HH17survey,HHprepar}, which  
aims to extend concepts in {\em Discrete Convex Analysis (DCA)}
(Murota~\cite{MurotaBook}) 
to tackle further various combinatorial optimization problems 
beyond network flows, matroids, and submodular functions.  
We will utilize these theories in a self-contained way.

In Section~\ref{sec:multiflow},  following~\cite{HHMPA}
we formulate the dual of our multiflow problem as a facility location problem on a tree. 
This formulation gives a fruitful 
combinatorial interpretation of the LP-dual problem~(\ref{eqn:LP-dual}), 
and brings a simple combinatorial algorithm to find 
a half-integral optimal multiflow from a given dual optimum, 
under a certain nondegeneracy assumption.
We will deal with a perturbed problem satisfying this nondegeneracy assumption.
Our goal is to solve this perturbed problem efficiently.
We will see that
the location problem is further formulated as
an optimization over a certain discrete structure, and
the objective is an {\em L-convex function on a Euclidean building} in the sense of~\cite{HHprepar}. 
This class of discrete convex functions shares many analogous properties with 
L-convex functions in DCA.
In particular, as in the case of DCA, 
there is a natural descent algorithm,  
called the {\em steepest descent algorithm}, 
to minimize our L-convex function $g$.
For each point $x$,  
the steepest descent algorithm 
chooses a point $y$ ({\em steepest direction}) 
from a {\em discrete neighborhood} of $x$ with smallest $g(y)$.
If $g(y) = g(x)$, then $x$ is guaranteed to be optimal.
Otherwise, i.e., $g(y) < g(x)$, replace $x$ by $y$, and repeat.

In Section~\ref{sec:algo}, 
we will implement
this conceptually simple algorithm.
We will prove that in our case a steepest direction at each point
can be found by solving one maximum submodular flow problem.
This part is the heart of our analysis.
As a consequence, we obtain an algorithm in a simple form as follows:
\begin{itemize}
	\item[1.] From a dual solution ({\em potential}) $x = (p,r)$, construct and solve an instance of 
	the maximum submodular flow problem.
	\item[2.] If the minimal minimum cut consists only of  
	the source, then $x$ is optimal, and an optimal multiflow is constructed 
	from any maximum submodular flow. Otherwise the minimal minimum cut gives a steepest direction $y$ of 
	the neighborhood at~$x$. 
	Replace $x$ by $y$, and go to 1.  
\end{itemize}
Our maximum submodular flow problem
is defined by a disjoint sum of submodular functions on $6$-element sets. 
This enables us to compute the exchange capacity in constant time.
Moreover, the number of iterations is estimated by the {\em geodesic descent property} 
(Theorem~\ref{thm:bound})
of the steepest descent algorithm. This intriguing property says that 
a trajectory of the algorithm 
forms a geodesic to optimal solutions with respect to a certain $l_{\infty}$-metric on the domain. 
We know in advance 
the range where an optimum exists, 
and the diameter of the range is bounded by $O(m \log k)$ relative to the above metric.
Consequently 
the number of iterations is bounded by $O(m \log k)$.
%

It should be noted that our algorithm design includes 
an interesting new technique 
of reducing bisubmodularity to submodularity.
This technique and related arguments, 
including basics on submodularity, are summarized in Section~\ref{sec:pre}. 
Actually step 1 of the above algorithm is essentially 
the feasibility check of a {\em bisubmodular flow} problem, that is, finding
a (fractional) bidirected flow with the flow-boundary constrained to a bisubmodular polyhedron.
This seemingly natural class of problems has not been well-studied so far. 
On the other hand, 
it is well-known that (fractional) bidirected flows are easily manipulated 
by ordinary flows in a {\em skew-symmetric network}
obtained by doubling nodes and edges.
We generalize this doubling construction to bisubmodular functions.
We give a condition for a bisubmodular function $f$ to be extended to a submodular function $f'$ on a larger set, 
so that the bisubmodular polyhedron of $f$ is a projection of the base polyhedron of $f'$. 
We show that 
a certain bisubmodular function on a 3-element set, 
which represents the flow-conservation and the node-capacity constraints on a node of degree 3, 
has such a submodular extension on a 6-element set.
Our bisubmodular flow problem is described 
by the disjoint sum of these bisubmodular functions, 
and can be reduced to the submodular flow problem as mentioned above.

The results of this paper is also outlined in an expository article \cite{HH17survey} on discrete convexity and algorithm design.

\section{Preliminaries}\label{sec:pre}

\paragraph{Notation.}
Let $\RR$, $\RR_+$, $\ZZ$, and $\ZZ_+$ 
denote the sets of reals, nonnegative reals, integers, and nonnegative integers, respectively.
The infinity element $\infty$ is treated as $x < \infty$ and $x + \infty = \infty$ for $x \in \RR$.
The set of all functions from a set $V$ to a set $R$ is denoted by $R^V$.
For a function $v \in \RR^V$ 
and a subset $X \subseteq V$, let $v(X)$ denote $\sum_{x \in X}v(x)$.
The function value $v(i)$ will also be denoted by $v_i$ if no confusion occurs.
For a (directed or undirected) graph $G =(V,E)$, an edge from $i$ to $j$ is denoted by $ij$.
For a subset $X$ of nodes, let $\delta X$ denote 
the set of all edges leaving $X$.
For an undirected graph $\varGamma$ with a specified edge-length, 
let $d = d_{\varGamma}$ denote the shortest path metric on the vertex set  
with respect to the edge-length.
In the following, graphs or networks are supposed to have no multiple edges and loops.

\paragraph{Signed set and transversal.} 
A {\em signed set} $U$ is the product $V \times \{+,-\}$ of a set $V$ and the sign $\{+, -\}$.
Elements $(i,+)$ and $(i,-)$ of $U$ are simply denoted by $i^+$ and $i^-$, respectively.
The {\em signed extension} of a set $V$ is defined as the signed set $V \times \{+,-\}$ 
and is denoted by $V^{\pm}$.
For $Y \subseteq V$, 
let $Y^+ := \{i^+ \mid i \in Y\}$ and $Y^- := \{i^- \mid i \in Y\}$.
Also for $U \subseteq V^{\pm}$, 
let $U^+ := \{i^+ \mid i^+ \in U\}$ and $U^- := \{i^- \mid i^- \in U\}$. 
A subset $X$ of $V^{\pm}$ is called a {\em transversal} if 
$|X \cap \{i^+, i^-\}| \leq 1$ for all $i \in V$, and is called 
a {\em co-transversal} if $|X \cap \{i^+, i^-\}| \geq 1$ for all $i \in V$.
For $X \subseteq V^{\pm}$
let $\underline X$ denote the transversal obtained from $X$ by deleting all $\{i^+,i^-\}$ with $\{ i^+,i^-\} \subseteq X$,
and let $\overline X$ denote the co-transversal obtained from $X$ 
by adding all $\{i^+,i^-\}$ with $\{i^+,i^-\} \cap X = \emptyset$.
For $u \in V^{\pm}$, define $\bar u$ by 
$\bar u := i^+$ if $u = i^-$ and $\bar u := i^-$ if $u = i^+$.

\paragraph{Skew-symmetric network.}
A {\em skew-symmetric network} (see e.g.,~\cite{GK96}) is a directed network on a signed set 
such that edge $uv$ exists if and only if edge $\bar v \bar u$ exists, and
the (lower and upper) capacities of edges $uv$ and $\bar v \bar u$ are the same.
A skew-symmetric network is often useful for dealing with 
problems in undirected graphs.

\subsection{Submodular flow}
Here we summarize basics on submodular functions and submodular flows; 
see \cite{FrankBook,FujiBook,FI00survey} for further details.
A {\em submodular function} on a set $V$ is a function $\rho$ defined on $2^{V}$ satisfying
$
\rho(X) + \rho(Y) \geq \rho(X \cap Y) + \rho(X \cup Y)$ for $X,Y \subseteq V$.
Let $\rho$ be  a submodular function on $V$ with $\rho(\emptyset) = 0$.
The {\em base polyhedron} ${\cal B}(\rho)$ is defined as the set of all vectors
$x \in \RR^{V}$ satisfying
$x(X) \leq \rho(X)$ for $X \subseteq V$ and  $x(V) = \rho(V)$.
For $x \in {\cal B}(\rho)$ and a pair $(i,j)$ of distinct elements of $V$, 
the {\em exchange capacity} $\kappa(x; i, j)$ at $x$ is defined by
\[
\kappa(x; i, j) := \max \{ \alpha \in \RR_+ \mid x + \alpha( \chi_i - \chi_j) \in {\cal B}(\rho) \}, 
\] 
where $\chi_i$ is the $i$-th unit vector defined by $\chi_i(j) := 1$ if $i=j$ and $\chi_i(j) := 0$ otherwise.

We next introduce submodular flows. 
Let $N$ be a directed network with vertex set $V$, edge set $A$, edge-capacity $c: A \to \RR_+$, 
and terminals $s,t \in V$.
The set of nonterminal nodes is denoted by $U (:= V \setminus \{s,t\})$.
We are given a submodular function $\rho :2^{U} \to \RR$ 
with $\rho(\emptyset) = \rho(U) = 0$.
For a function $\varphi: A \to \RR$, let $\nabla \varphi$ denote the function on $V$ defined by
\[
\nabla \varphi (i) := \sum \{  \varphi(e) \mid \mbox{$e \in A$: $e$ enters $i$}\}  
- \sum \{  \varphi(e) \mid \mbox{$e \in A$: $e$ leaves $i$}\}.
\]
Namely $\nabla \varphi(i)$ represents the excess of $\varphi$ at node $i$.
An {\em $(s, t)$-flow}, or simply, a {\em flow} 
is a function $\varphi: A \to \RR_+$ satisfying  
\begin{equation*}
0 \leq \varphi(e) \leq c(e) \quad (e \in A), \quad \nabla \varphi |_{U} \in {\cal B}(\rho),
\end{equation*}
where $(\cdot)|_{U}$ means the restriction to $U$.
The flow-value of a flow $\varphi$ is defined 
as $\nabla \varphi(t)$
By $\rho(U) = 0$, it holds $\nabla \varphi(t) = - \nabla \varphi(s)$.
An {\em $(s,t)$-cut} $X$ is a subset of nodes containing $s$ and not containing $t$.
The {\em cut capacity} of $X$ is defined as 
\begin{equation}\label{eqn:capacity}
c( \delta X) + \rho(X \setminus \{s\}).
\end{equation}
An $(s,t)$-flow is called {\em maximum} if it has the maximum flow-value, 
and an $(s,t)$-cut is called {\em minimum} if it has the minimum capacity.
\begin{Thm}[{see \cite[Theorem 5.11]{FujiBook}}]
	The maximum flow-value of an $(s,t)$-flow is equal to the minimum cut-capacity 
	of an $(s,t)$-cut $X$.
	If $c$ and $\rho$ are both integer-valued, 
	then there exists an integer-valued maximum flow.
	For any maximum flow $\varphi$, 
	the set of nodes reachable from $s$ in the residual network $N_{\varphi}$ of $\varphi$
	is the unique minimal minimum $(s,t)$-cut.
\end{Thm}
Here the {\em residual network} $N_{\varphi}$ of $\varphi$ 
is a directed network  on $V$ constructed as follows:
For each edge $ij$ in $N$ with $\varphi(ij) < c(ij)$, add an edge $ij$ to $N_{\varphi}$.
For each edge $ij$ in  $N$ with $0< \varphi(ij)$, add an edge $ji$ to $N_{\varphi}$. 
For each pair of distinct nodes $i,j$ in $U$ with $\kappa(\nabla \varphi|_{U};i,j) > 0$, add an edge $ij$ to $N_{\varphi}$.

There are several combinatorial polynomial time algorithms 
for computing an integral maximum submodular flow, 
under the assumption that an oracle for computing the exchange capacity is available.
They are designed by extending 
existing max-flow algorithms; see the survey~\cite{FI00survey} for further details.

We note one basic property for the case where the network is skew-symmetric.
\begin{Lem}\label{lem:minimal_minimum}
	Suppose that $V$ is a signed set with $t = \bar s$, and $N$ is skew-symmetric. 
	If $\rho$ satisfies $\rho(\underline{X}) \leq \rho(X)$ for all $X \subseteq U$, 
	then the unique minimal minimum $(s,t)$-cut is a transversal. 
\end{Lem}
\begin{proof}
	By the assumption for $\rho$, 
	it suffices to show $c(\delta \underline{X}) \leq c(\delta X)$. 
	Suppose that some edge $uv (\not \in \delta X)$ appears in $\delta \underline{X}$.
    In this case, it holds $u \in \underline{X} \subseteq X$ and $v \in X \setminus \underline{X}$.
	This means that $\{v, \bar v\} \subseteq X$ and $\bar u \not \in X$.
	Hence edge $\bar v \bar u \in \delta X$ (of the same capacity) does not appear in 
	$\delta \underline{X}$.
	Consequently the cut-capacity does not increase.
\end{proof}

\subsection{Submodular extension}\label{subsec:ext}
Here we introduce a method of reducing bisubmodularity 
to submodularity, which will play a key role in our algorithm; 
see the lower right of Figure~\ref{fig:outline}.
Let $V := \{1,2,\ldots,n\}$.
Let $3^{V}$ denote the set of ordered pairs of disjoint subsets of $V$. 
For a function $h$ on $3^V$, the polyhedron ${\cal D}(h)$ in $\RR^V$ 
is defined to be the set of all vectors $z \in \RR^V$ satisfying
\begin{equation}\label{eqn:D(h)}
z(Y) - z(Z) \leq h(Y,Z) \quad ((Y,Z) \in 3^V).
\end{equation}
If $h$ is a bisubmodular function\footnote{
A function $h$ on $3^V$ is called {\em bisubmodular} if it satisfies
	\[
	h(X,Y) + h(X',Y') \geq h(X \cap X', Y \cap Y') + h((X\cup X') \setminus (Y \cup Y'), (Y \cup Y') \setminus (X\cup X')) \quad ((X,Y),(X',Y') \in 3^V).
	\]
},
then ${\cal D}(h)$ is known as a bisubmodular polyhedron.
We are interested in the case where ${\cal D}(h)$ 
is a projection of the base polyhedron of some submodular function.
A representative of such polyhedra is the polyhedron of all flow-boundaries of 
a bidirected network; see~\cite{AFN97}.  

Let $V^{\pm} = \{1^+,2^+,\ldots,n^+,1^-,2^-,\ldots,n^-\}$ be the signed extension of $V$.
For a function $h$ on $3^V$, 
a {\em normal submodular extension} of $h$
is a submodular function $\rho$ on $2^{V^{\pm}}$ with $\rho(\emptyset) =0$ satisfying
\begin{eqnarray}
&& \rho(Y^+ \cup Z^-) = h(Y,Z) \quad ((Y,Z) \in 3^V), \label{eqn:extension1} \\ 
&& \rho(\overline{X}) = \rho(\underline{X}) \leq \rho(X) \quad (X \in 2^{V^{\pm}}).  \label{eqn:extension2} 
\end{eqnarray}
Define a map (projection) $\phi: \RR^{V^\pm} \to \RR^V$ by
\begin{equation}\label{eqn:def_phi}
(\phi(x))(i)  := \frac{x(i^+) - x(i^-)}{2} \quad (x \in \RR^{V^{\pm}}, i \in V).
\end{equation}
\begin{Lem}\label{lem:projection}
	Let $h$ be a function on $3^V$ and 
	let $\rho$ be a normal submodular extension of $h$.
	Then it holds $\phi({\cal B}(\rho)) = {\cal D}(h)$.
\end{Lem}
\begin{proof}
	Note that $\rho(V^{\pm}) = \rho(\emptyset) = 0$.
	We first show that $\phi({\cal B}(\rho)) \subseteq {\cal D}(h)$. 
	Take an arbitrary $x$ in ${\cal B}(\rho)$.
	For $(Y,Z) \in 3^{V}$, let $X := Y^+ \cup Z^-$. Note that $X = \underline{X}$ and $\rho(X) = \rho(\underline{X}) = \rho(\overline{X})$. Then we have 
	\[
	x(X) \leq \rho(X),\ x(\overline{X}) \leq \rho(\overline{X}) = \rho(X),\ x(V^{\pm}) = \rho(V^{\pm}) = 0.
	\]
	Hence we have
	$
	(x(X) + x(\overline{X}) -  x(V^{\pm}))/2 \leq (\rho(X) + \rho(\overline{X}) - \rho(V^{\pm}))/2 = \rho(X) = h(Y,Z).
	$
	Also we have
	\begin{eqnarray*}
	 (x(X) + x(\overline{X}) -  x(V^{\pm}))/2 & =& x(Y^+) + x(Z^-) 
		-  (x(Y^+ \cup Z^+) + x(Y^- \cup Z^-))/2 \\
		& = & (x(Y^+) - x(Y^-))/2 -  (x(Z^+) - x(Z^-))/2 \\ 
		& = & (\phi(x))(Y)  -  (\phi(x))(Z).
	\end{eqnarray*}
	Hence $\phi(x)$ belongs to ${\cal D}(h)$.	
	
	Next we show the converse.
	Take an arbitrary $z \in {\cal D}(h)$. 
	Define a vector $x$ in $\RR^{V^{\pm}}$ by 
	$x(i^+) := z(i)$ and $x(i^-) := - z(i)$ for $i \in V$.
	Obviously $\phi(x) = z$. 	
	It suffices to show that $x$ belongs to ${\cal B}(\rho)$.
	Since $x(V^{\pm}) = 0$, we have $x(V^{\pm}) = 0 = \rho(\emptyset) = \rho(V^{\pm})$.
	For $X \subseteq V^{\pm}$, we have
	\begin{equation*}
	x(X) = x(\underline{X})= \sum_{i: i^+ \in \underline{X}} z(i) -\sum_{i: i^- \in \underline{X}} z(i) 
	\leq h(\underline{X}^+, \underline{X}^-) = \rho(\underline{X}) \leq \rho(X). 
	\end{equation*}
    Thus $x \in {\cal B}(\rho)$, and hence ${\cal D}(h) \subseteq \phi({\cal B}(\rho))$.
\end{proof}
If $h:3^{V} \to \RR$ has a normal submodular extension, 
then $h$ is necessarily a bisubmodular function. 
Not all bisubmodular functions admit submodular extensions (Y. Iwamasa 2015).

We consider a special bisubmodular function on $3$-element set $\{1,2,3\}$, 
which plays a key role in Section~\ref{sec:multiflow}.
For $b \geq 0$,
let $\mDelta_b$ be the function on $3^{\{1,2,3\}}$ defined by
\begin{equation}
\mDelta_b (Y,Z) := \left\{ 
\begin{array}{ll}
2b & {\rm if}\ |Y| \geq 2, \\
b & {\rm if}\ |Y| =1, |Z| \leq 1, \\
0 & {\rm otherwise}\ (Y =\emptyset \ {\rm or}\ |Z| \geq 2),
\end{array}
\right. \quad (Y,Z) \in 3^V.
\end{equation}
\begin{Lem}\label{lem:flow_conservation}
	The polyhedron ${\cal D}(\mDelta_b)$ is the set of nonnegative vectors $z \in \RR^{\{1,2,3\}}_+$ satisfying
	\begin{eqnarray*}
	&& \ \  z(1) + z(2) + z(3) \leq 2 b, \\
	&& \ \  z(1) - z(2) - z(3) \leq 0, \\
	&& - z(1) + z(2) - z(3) \leq 0, \\
	&& - z(1) - z(2) + z(3) \leq 0.
	\end{eqnarray*}	
\end{Lem}
\begin{proof}
	Observe that
	these inequalities appear in (\ref{eqn:D(h)}).
	So it suffices to show that inequalities~(\ref{eqn:D(h)}) 
	are derived from the above inequalities.
	This is a routine verification.
	For example, 
	$z(1) + z(2) - z(3) \leq \varDelta_b(\{1,2\},\{3\}) = 2b$
	is obtained by
	adding $z(1) + z(2) + z(3) \leq 2b$ and $- z(3) \leq 0$.
	Also $z(1) \leq \varDelta_b(\{1\}) = b$ is implied by $z(1)+ z(2)+ z(3) \leq 2b$ and $z(1)- z(2) - z(3) \leq 0$, and
	$z(1) - z(2) \leq \varDelta_b(\{1\},\{2\}) = b$ is 
	implied by $z(1) - z(2) -z(3) \leq 0$ and $z(3) \leq b$.	 
\end{proof}
The polyhedron ${\cal D}(\mDelta_b)$ is a simplex 
with vertices $(0,0,0)$, $(b,b,0)$, $(b,0,b)$ and $(0,b,b)$.
We will see in Section~\ref{sec:multiflow} that 
${\cal D}(\mDelta_b)$
represents 
the flow-conservation law and the node-capacity constraint
on a node of degree $3$. 
%
This bisubmodular function $\mDelta_b$ has a normal submodular extension. 
The following example  was 
found by Yuni Iwamasa via computer calculation.
Classify subsets $X \subseteq \{1^+,2^+,3^+,1^-,2^-,3^-\}$ 
into the following six types:
\begin{description}
	\item[{\rm type 1:}] $|X^+| \geq 2$ and  $|X^-| \leq 1$.
	\item[{\rm type 2:}] $X^+ = \{i^+\}$ and $X^- =\{ 1^-,2^-,3^-\} \setminus \{i^-\}$ for some $i \in \{1,2,3\}$. 
	\item[{\rm type 3:}] $X \subseteq \{1^-,2^-,3^-\}$ or $\{1^-,2^-,3^-\} \subseteq X$. 
	\item[{\rm type 4:}] $|X^+| = 2$ and  $|X^-| = 2$.
	\item[{\rm type 5:}]  $X^+ = \{i^+\}$, $|X^-| \leq 2$, and $X^- \neq \{1^-,2^-,3^- \} \setminus \{i^-\}$
	for some $i \in \{1,2,3\}$.
	\item[{\rm type 6:}] $X = \{1^+,2^+,3^+,1^-,2^-,3^-\} \setminus \{i^-\}$ for some $i \in \{1,2,3\}$.
\end{description}
Define $\mDelta^*_b: 2^{\{1^+,2^+,3^+,1^-,2^-,3^-\}} \to \RR$ by 
\begin{equation}\label{eqn:delta*}
\mDelta^*_b(X)
:= \left\{
\begin{array}{ll}
2b & {\rm if}\ \mbox{$X$: type 1},\\
0 & {\rm if}\ \mbox{$X$: type 2 or 3},\\
b & {\rm otherwise}\ \mbox{($X$: type 4, 5, or 6)}.
\end{array} \right.
\end{equation}
\begin{Lem}\label{lem:Delta*}
	$\mDelta^*_b$ is a normal submodular extension of $\mDelta_b$.
\end{Lem}
\begin{proof}
It suffices to consider the case of $b=1$; 
we denote $\mDelta_1$ and $\mDelta_1^*$ by $\mDelta$ and $\mDelta^*$, respectively.

First we show (\ref{eqn:extension1}).
For $(Y,Z) \in 3^{\{1,2,3\}}$, let $X := Y^+ \cup Z^-$. 
Then $X^+ = Y^+$ and $X^- = Z^-$.
If $|Y|  = |X^+|\geq 2$, 
then $|Z| = |X^-| \leq 1$, and $X$ is of type 1; 
hence $\mDelta^*(X) = 2 = \mDelta(Y,Z)$.
If $|Y| = |X^+| = 1$ and $|Z| = |X^-| \leq 1$, 
then $X$ is of type 5, and hence 
$\mDelta^*(X) = 1 = \mDelta(Y,Z)$.
If $Y$ is empty,
then $X \subseteq \{1^-,2^-,3^-\}$, and $X$ is type 3; 
hence $\mDelta^*(X) = 0 = \mDelta(Y,Z)$.
If $|Z| = |X^-| \geq 2$ and $Y \neq \emptyset$, 
then $|Y| = |X^+| = 1$, and $X$ is of type 2; 
hence $\mDelta^*(X) = 0 = \mDelta(Y,Z)$.

Second we show (\ref{eqn:extension2}).
It suffices to show that $\mDelta^*(X) = \mDelta^*(\overline{X})$ holds
for any transversal $X$, and that $\mDelta^*(\underline{X}) \leq \mDelta^*(X)$ holds 
for any $X$ that is neither a transversal nor a co-transversal.
The former property follows from 
$\mDelta^*(\{1^+\}) = 1 = \mDelta^*(\{1^+,2^+,3^+,2^-,3^-\})$,
$\mDelta^*(\{1^-\}) = 0 = \mDelta^*(\{2^+,3^+,1^-, 2^-,3^-\})$,
$\mDelta^*(\{1^+,2^+\}) = 2 = \mDelta^*(\{1^+,2^+,3^+,3^-\})$,
$\mDelta^*(\{1^+,2^-\}) = 1 = \mDelta^*(\{1^+,3^+, 2^-,3^-\})$, and
$\mDelta^*(\{1^-,2^-\}) = 0 = \mDelta^*(\{3^+,1^-, 2^-,3^-\})$.
The latter property follows from
$\mDelta^*(\{2^+,2^-\}) = 1 > 0 = \mDelta^*(\emptyset)$,
$\mDelta^*(\{1^+,2^+,2^-\}) = 2 > 1 = \mDelta^*(\{1^+\})$, and
$\mDelta^*(\{2^+,1^-,2^-\}) = 1 > 0 = \mDelta^*(\{1^-\})$.

Finally we show the submodularity of $\mDelta^*$.
Take $X,Y \subseteq \{1^+,2^+,3^+,1^-,2^-,3^-\}$.
We can assume that $X \not \subseteq Y$ and $Y \not \subseteq X$.

Case 1: $X$ is of type 6.
In this case, $X \cup Y$ is the whole set, and is of type 3.
Therefore it suffices to consider the case where $X \cap Y$ is of type 1 and $Y$ is not of type 1.
Necessarily $Y$ is of type 4 or 6.
Thus submodular inequality $1 + 1 \geq 2 + 0$ holds.  

Case 2: $X$ is of type 2 with $1^+ \in X$.
If $Y$ is also of type 2, then $X \cup Y (\supseteq \{1^-,2^-,3^-\})$ is of type 3, 
$X \cap Y (\subseteq \{1^-,2^-,3^-\})$ is also of type 3.
If $Y$ is of type 3 with $Y \subseteq \{1^-,2^-,3^-\}$, 
then both $X \cap Y$ and $X \cup Y$ are of type 3.
If $Y$ is of type 3 with $Y \supseteq \{1^-,2^-,3^-\}$,
then $Y$ cannot have $1^+$, and thus both $X \cap Y$ and $X \cup Y$ are of type 3.
In these cases, submodularity $(0+0 \geq 0+0)$ holds.
Thus we may assume that $Y$ is of type 1,4, or 5.
Observe that neither $X \cap Y$ nor $X \cup Y$ is of type 1.
We may assume that $Y$ contains $1^+$ and does not contain $1^-$;
otherwise $X \cap Y$ or $X \cup Y$ is of type 3, and submodularity holds.
Necessarily $Y$ is of type 1.
Then $X \cap Y$ is of type 5, and $X \cup Y$ is of type 4 or 6;
submodularity $(0+2 \geq 1 + 1)$ holds.

Case 3: $X$ is of type 3 and $Y$ is not of type 2.
If $Y$ is also of type 3, 
then both $X \cap Y$ and $X \cup Y$ is of type 3; 
submodularity holds.
Since one of $X \cap Y$ and $X \cup Y$ is of type 3, 
it suffices to consider the case where $X \cap Y$ or $X \cup Y$ is of type 1.
We show that $Y$ is also type 1.
If $X \cap Y$ is of type 1, 
then $X = \{1^+,2^+,3^+,1^-,2^-,3^-\} \setminus \{i^+\}$, and necessarily $Y$ is of type 1.
If $X \cup Y$ is of type 1,
then $X = \{i^-\}$, and $Y$ has at least two elements in $\{1^+,2^+,3^+\}$~(type 1).

Case 4: $X$ is of type 4 or 5, and $Y$ is type 1, 4, or 5. 
Suppose that $X$ is of type 4.
Then $X \cup Y$ is not type 1.
We may consider the case where $X \cap Y$ is of type 1 and $Y$ is not of type~1.
Then $Y$ contains $X^+$. Thus $Y$ is of type 4, and $X \cup Y$ necessarily contains $\{1^-,2^-,3^-\}$ (type 3); 
submodularity ($1 + 1 \geq 2 + 0$) holds.
Suppose that $X$ is of type 5 with $1^+ \in X$.
Then $X \cap Y$ is not type 1.
We may consider the case where $X \cup Y$ is of type 1.
If $Y$ does not contain $1^+$, then $X \cap Y$ is of type 3; submodularity ($1 + 1 \geq 2 + 0$) holds.
Thus $|Y^+| \geq 2$ and $|Y^-| = 0$ or $1$; $Y$ is of type 1.
The intersection $X \cap Y$ is of type 5; thus $1+ 2 \geq 2+1$ holds.    
\end{proof}
%

\section{Node-capacitated multiflow}\label{sec:multiflow}

In this section, we introduce a combinatorial 
duality theory, developed by~\cite{HHMPA},  for a class of node-capacitated multiflow problems. 
We consider the following multiflow problem. 
Now assume that network $N$ also has a nonnegative edge-cost $a: E \to \RR_+$; 
so the network is a 5-tuple $(V,E,S,c,a)$.
For a multiflow $f = ({\cal P},\lambda)$, the total flow-values on node $i$ and edge $e$ are denoted by
$f(i) :=  \sum_{P \in {\cal P}:\  i \in V(P)} \lambda(P)$ and $f(e) :=  \sum_{P \in {\cal P}:\  e \in E(P)} \lambda(P)$, respectively.
The {\em cost} $a(f)$ is defined by
\begin{equation*}
a(f) := \sum_{e \in E} a(e) f(e). 
\end{equation*}
Next we define the value of a multiflow.
A {\em tree-embedding} ${\cal E} = (\mGamma, \{q_s\}_{s \in S})$ 
is a pair of a tree $\mGamma$ and a family $\{q_s\}_{s \in S}$ 
of vertices of $\mGamma$ indexed by terminal set $S$.
The {\em ${\cal E}$-value} $v_{\cal E}(f)$ of a multiflow $f = ({\cal P}, \lambda)$ is defined by
\[
v_{\cal E}(f) := \sum_{P \in {\cal P}} d(q_{s_P},q_{t_P}) \lambda(P),
\]
where $s_P,t_P$ denote the ends of an $S$-path $P$, 
and $d = d_{\varGamma}$ denotes the shortest path metric of $\varGamma$ with respect to unit edge-length.
We are now ready to define our multiflow problem.
An instance of the problem is a pair of a network $N = (V,E,S,c,a)$ 
and a tree-embedding ${\cal E} = (\mGamma, \{p_s\}_{s \in S})$, 
and the task is to find a multiflow $f$ 
that maximizes $v_{\cal E}(f) - a(f)$.

This somewhat artificial formulation turns out to be useful, 
and actually generalizes the original problem.
Indeed, take $\mGamma$ as a star with $|S|$ leaves $v_s$ $(s \in S)$, 
let ${\cal E} := (\mGamma, \{v_s\}_{s \in S})$, and let $a(e) := 0$ for each edge $e$.
Then $v_{\cal E}(f) - a(f)$ is twice the total flow-value of $f$.

In Section~\ref{subsec:duality}, we deal with the left part in Figure~\ref{fig:outline}.
We present a combinatorial duality theorem and an optimality criterion.
We introduce a nondegeneracy concept of the problem, 
and give an algorithm to 
find a half-integral optimal multiflow from a dual optimum under 
the nondegeneracy assumption.
We also explain how to reduce the original problem 
to a nondegenerate problem.
In Section~\ref{subsec:dc}, 
we deal with the upper right part in Figure~\ref{fig:outline}.
We show that our dual objective can be viewed as
an L-convex function on a certain graph structure, and 
present the steepest descent algorithm (SDA) 
to minimize L-convex functions and its iteration bound.
\subsection{Duality}\label{subsec:duality}
Let a pair of $N= (V,E,S,c,a)$ and ${\cal E} = (\mGamma, \{q_s\}_{s \in S})$ be an instance of the problem. 
We may assume that there is no edge connecting terminals.
The vertex set of $\mGamma$ is also denoted by $\mGamma$ (instead of $V(\mGamma)$).
Let $\mGamma^*$ denote the edge-subdivision of $\mGamma$, 
where $\mGamma \subseteq \mGamma^*$,
the edge-length of $\mGamma^*$ is defined as $1/2$ uniformly, and
the shortest path metric $d_{\mGamma^*}$ is also denoted by $d$.

A pair $(p,r)$ of a tree-valued function $p: V \to \mGamma^*$ 
and a nonnegative half-integer-valued function $r: V \to \ZZ_+/2$ 
is called a {\em potential} if 
it satisfies the following conditions:
\begin{itemize}
	\item[(p1)] For each node $i$,  $r(i)$ is an integer if and only if $p(i)$ belongs to $\mGamma$.
	\item[(p2)] For each edge $ij$, it holds $d(p(i), p(j))- r(i) - r(j) \leq a(ij)$.
	\item[(p3)] For each terminal $s$, it holds $(p(s),r(s)) = (q_s, 0)$.
\end{itemize}
Then the following min-max formula and optimality criterion hold:
\begin{Thm}[\cite{HHMPA}]\label{thm:duality}
	Suppose that $a$ is even-valued.
	The maximum of $v_{\cal E}(f) - a(f)$ over all multiflows $f$ is equal to the minimum of
	$
	\sum_{i \in V \setminus S} 2 c(i) r(i)
	$
	over all potentials $(p,r)$. 
\end{Thm}
\begin{Lem}[\cite{HHMPA}]\label{lem:optimality'}
	Suppose that $a$ is even-valued.
	A multiflow $f = ({\cal P}, \lambda)$ and a potential $(p,r)$ 
	are both optimal if and only if they satisfy the following conditions:
	\begin{itemize}
		\item[{\rm (o1)}] For each path $P$ in ${\cal P}$ with $\lambda(P) > 0$, it holds 
		$\displaystyle
		d(q_{s_P},q_{t_P}) = \sum_{ij \in E(P)} d(p(i), p(j)). 
		$
		\item[{\rm (o2)}] For each edge $ij$ with $f(ij) > 0$, it holds
		$
		d(p(i),p(j)) - r(i) - r(j) = a(ij).
		$ 
		\item[{\rm (o3)}] For each nonterminal node $i$ with $r(i) > 0$, it holds $f(i) = c(i)$.
	\end{itemize}	
\end{Lem}
We will use the if part (and the weak duality in Theorem~\ref{thm:duality}) only, 
which is proved for completeness.
\begin{proof}
	(If part). 
	For any multiflow $f = ({\cal P},\lambda)$ and any potential $(p,r)$,
	 the difference $\sum_{i \in V \setminus S} 2 c_i r_i  - ( v_{\cal E}(f)- a(f))$ is equal to
	\begin{eqnarray}\label{eqn:weak_duality}
	  & & \sum_{i \in V \setminus S} 2 (c_i - f_i) r_i + \sum_{ij \in E} f_{ij}( a_{ij} - d(p_i,p_j) + r_i + r_j) \nonumber \\
	 && \quad + \sum_{P \in {\cal P}} \lambda(P) \left( 
	  \sum_{ij \in E(P)} d(p_i, p_j) - d(q_{s_P}, q_{t_P})    
	  \right)  \geq 0,
	\end{eqnarray}
	where we use 
	\begin{eqnarray*}
	&& \sum_{ij \in E} f_{ij} d(p_i,p_j) = \sum_{P \in {\cal P}} \lambda(P) \sum_{ij \in E(P)} d(p_i,p_j), \\
	&& \sum_{i \in V \setminus S} 2 f_{i} r_i = \sum_{ij \in E} f_{ij}(r_i+r_j).
	\end{eqnarray*}
	Thus, if $f$ and $(p,r)$ satisfy conditions (o1), (o2), and (o3), 
	then the equality holds in (\ref{eqn:weak_duality}), 
	and both $f$ and $(p,r)$ are optimal.
\end{proof}

\paragraph{Nondegenerate case.}
An instance $(N, {\cal E})$
is said to be {\em nondegenerate} if the edge-cost $a$ is positive even-valued 
and the degree of each node in $\mGamma$ is at most $3$.
Suppose that $(N, {\cal E})$ is nondegenerate.
We further assume, for notational simplicity, that tree $\mGamma$ 
has no vertex of degree one 
(by attaching paths of infinite length).
Let $\mGamma_2$ and $\mGamma_3$ denote the sets of vertices of $\mGamma$ 
with degree 2 and 3, respectively.
For a vertex $v$ in $\mGamma$, the neighbors of $v$ in $\mGamma$
are denoted by $v_{\to 1}, v_{\to 2}$ if $v \in \mGamma_2$ and 
$v_{\to 1}, v_{\to 2}, v_{\to 3}$ if $v \in \mGamma_3$.
Consider the edge-subdivision $\mGamma^*$ of $\mGamma$.
%
For vertex $v \in \mGamma^*$, the neighbors of $v$ in $\mGamma^*$
are denoted by $v_{\to^* 1}, v_{\to^* 2}$ if $v \in \mGamma_2$ or $v \in \mGamma^* \setminus \mGamma$, and $v_{\to^*1}, v_{\to^* 2}, v_{\to^* 3}$ if $v \in \mGamma_3$.
Let $\mGamma^*_{v, k}$ denote
the connected component of $\mGamma^* - v$ containing $v_{\to^* k}$.

We are going to characterize the flow support of an optimal multiflow.
Let $(p,r)$ be a potential.
Motivated by (o2), 
define the edge subset $E_{p,r}$ by
\begin{equation*}
E_{p,r} := \{ij \in E \mid d(p(i),p(j)) - r(i) - r(j) = a(ij)\}.
\end{equation*}
For a nonterminal node $i$, 
let $\delta_{p, k}(i) (= \delta_{p,r,k}(i)) $ denote the set of edges $ij \in E_{p,r}$
with $p(j) \in \mGamma^*_{p(i), k}$.
Since each edge cost $a(ij)$ is positive, it holds
$p(i) \neq p(j)$ for $ij \in E_{p,r}$.
Thus $\delta_{p, k}(i)$ for $k = 1,2,3$ (or $k=1,2$) partition the set $\delta \{i\}$ 
of all edges in $E_{p,r}$ incident to $i$.

A {\em $(p,r)$-admissible support} is a function $\zeta: E_{p,r} \to \RR_+$ 
satisfying the following conditions, where we use the notational convention $\zeta(\delta_{p,k}(i)) := \sum_{e \in \delta_{p,k}(i)} \zeta(e)$:
\begin{itemize}
	\item[(a1)] For each nonterminal node $i$ with $p(i) \not \in \mGamma_3$, it holds
	$
	\zeta( \delta_{p,1}(i)) = \zeta(\delta_{p,2}(i)) \leq c(i).   
	$
	\item[(a2)] For each nonterminal node $i$ with $p(i) \in \mGamma_3$, it holds
	\begin{eqnarray*}
	&& \ \  \zeta( \delta_{p, 1}(i)) + \zeta( \delta_{p, 2}(i)) + \zeta( \delta_{p, 3}(i)) \leq 2 c(i), \\ 
	&& \ \  \zeta( \delta_{p, 1}(i)) - \zeta( \delta_{p, 2}(i)) - \zeta( \delta_{p, 3}(i)) \leq 0, \\ 
	&& - \zeta( \delta_{p, 1}(i)) + \zeta( \delta_{p, 2}(i)) - \zeta( \delta_{p, 3}(i)) \leq 0, \\ 
	&& - \zeta( \delta_{p, 1}(i)) - \zeta( \delta_{p, 2}(i)) + \zeta( \delta_{p, 3}(i)) \leq 0.
	\end{eqnarray*}
	\item[(a3)] For each nonterminal node $i$ with $r(i) > 0$, it holds
	$
	\zeta (\delta \{i\}) = 2c(i).
	$
	\item[(a4)] For each edge $e$, $\zeta(e)$ is a half-integer, and 
	for each nonterminal node $i$, $\zeta(\delta\{i\})$ is an integer. 
\end{itemize}
It is not difficult to see from Lemma~\ref{lem:optimality'} that for any half-integral optimal multiflow $f$, 
the flow-support $\zeta$ of $f$, defined by $\zeta(e) := f(e)$, 
is a $(p,r)$-admissible support. 
Indeed, the inequality in (a1) and the first inequality in (a2) 
are nothing but the capacity constraints.
Also (a3) corresponds to (o3).
The equality in (a1) and 
the last three inequalities in (a2) come from (o1), 
which says that a flow entering $i$ from $\delta_{p,k}(i)$ goes out 
through $\delta_{p,k'}(i)$ with $k' \neq k$.
Furthermore, the converse also holds. 
\begin{Lem}[\cite{HHMPA}]\label{lem:mn}
	Let $(p,r)$ be a potential. 
	If a $(p,r)$-admissible support $\zeta$ exists and is given, 
	then $(p,r)$ is optimal and
	a half-integral optimal multiflow is obtained in $O(n m)$ time.
\end{Lem}
Thus our problem is to find a potential $(p,r)$ 
such  that a $(p,r)$-admissible support exists.
Observe that a $(p,r)$-admissible support 
is viewed as an edge-weight $\zeta$ whose degree vector 
$\zeta(\delta_{p,k}(i))$ $( i \in V \setminus S, k=1,2,3)$
belongs to a bisubmodular polyhedron described by $\mDelta_{c(i)}$.
Namely, finding a $(p,r)$-admissible support 
is a bisubmodular flow feasibility problem.
In Section~\ref{sec:algo}, 
by using submodular extension $\mDelta^*_{c(i)}$ (Section~\ref{subsec:ext})
we reduce this problem to a maximum submodular flow problem. 

An algorithm for Lemma~\ref{lem:mn} is the following.
\begin{description}
	\item[Algorithm 1:] Construction of an optimal multiflow from a $(p,r)$-admissible support.
	\item[Input:] A potential $(p,r)$ and a $(p,r)$-admissible support $\zeta$.
	\item[Output:] A half-integral optimal multiflow $f = ({\cal P},\lambda)$.
	\item[Step 0:] ${\cal P} = \emptyset$.
	\item[Step 1:] Choose a terminal $s$ and an edge $sj$ with $\zeta(sj) > 0$.
	If such a terminal does not exist, then $f = ({\cal P}, \lambda)$ is a half-integral optimal multiflow; stop.
	Otherwise let $j_0 \leftarrow s$, $j_1 \leftarrow j$,   $\mu \leftarrow \zeta(sj)$, $l \leftarrow 1$, and go to step 2.
	\item[Step 2:] If $j_l$ is a terminal, then add path $P = (j_0, j_1,\ldots, j_l)$ to ${\cal P}$ with flow-value $\lambda(P) := \mu$, 
	let $\zeta (e) \leftarrow \zeta(e) - \mu$ for each edge $e$ in $P$, and go to step 1.
	Otherwise go to step 3.
	\item[Step 3:]
	If $p(j_l) \not \in \varGamma_3$ and $j_{l-1}j_l \in \delta_{p, k}(j_l)$ for $k \in \{1,2\}$, then
	choose an edge $j_{l} j_{l+1}$ from $\delta_{p, k'} (j_l)$ with $k' \neq k$ and  $\zeta(j_{l} j_{l+1}) > 0$, and let $\mu \leftarrow \min \{ \mu, \zeta(j_lj_{l+1}) \}$.
	
	If $p(j_l) \in \varGamma_3$ and $j_{l-1}j_l \in \delta_{p, k}(j_l)$ for $k \in \{1,2,3\}$,
	then choose an edge $j_{l} j_{l+1}$ from $\delta_{p, k'} (j_l)$ 
	with  $k' \neq k$, $\zeta(j_{l} j_{l+1}) > 0$, and $\zeta(\delta_{p, k}(j_l)) +  \zeta(\delta_{p, k'}(j_l)) - \zeta(\delta_{p, k''}(j_l)) > 0$ for $k'' \in \{1,2,3\} \setminus \{k,k'\}$.
	Let
 	\[
	 \mu \leftarrow \min \left\{ \mu, \zeta(j_{l}j_{l+1}), \frac{\zeta(\delta_{p, k}(j_l)) +  \zeta(\delta_{p, k'}(j_l)) - \zeta(\delta_{p, k''}(j_l))}{2} \right\}.
	 \]
	 Let $l \leftarrow l+1$ and go to step 2.
\end{description}
This algorithm is essentially the proof of \cite[Lemma 4.5]{HHMPA}.
Let us sketch the correctness of the algorithm; we show that the resulting multiflow $f$ satisfies 
the conditions (o1),(o2), and (o3) in Lemma~\ref{lem:optimality'} with $(p,r)$.
The condition (o2) follows from $f(e) = 0$ for $e \in E \setminus E_{p,r}$.  
In step~3, we can always choose a required edge by (a1) and (a2).
Also $\zeta$ still satisfies the conditions (a1), (a2), and (a4), thanks to the way of the update.
Each produced path $(j_0, j_1,j_2,\ldots, j_m)$  satisfies
\[
d(p(j_{l-1}), p(j_l)) + d(p(j_l), p(j_{l+1})) = d(p(j_{l-1}), p(j_{l+1})) \quad (1 \leq l \leq m-1)
\]
since
$p(j_{l-1}) \in \mGamma^*_{p(j_l),k}$ and $p(j_{l+1}) \in \mGamma^*_{p(j_l),k'}$ 
for $k \neq k'$.
Also each $d(p(j_{l-1}), p(j_l))$ is positive (since $a$ is positive).
Since $\mGamma$ is a tree, we have $d(p(j_0), p(j_m)) = \sum_{l=1}^m d(p(j_{l-1}), p(j_l))$; 
see e.g., \cite[Lemma 3.9]{HH14extendable}.
Thus each produced path satisfies (o1), and has no repeated node.
By the same argument, every edge $e$ with $\zeta (e) > 0$
extends to an $S$-path consisting of edges $e'$ with $\zeta(e') > 0$
satisfying~(o1).  
This means that if no terminal $s$ is chosen in step 1, 
then $\zeta = 0$.
By (a3), the resulting multiflow $f$ satisfies (o3).
Notice that $\mu$ is a half-integer by (a4).
Hence $f$ is half-integral and optimal.
Once an $S$-path $P$ is obtained, 
$\zeta$ becomes zero on some edge, 
or $\zeta(\delta_{p,k}(i)) + \zeta(\delta_{p,k'}(i)) - \zeta(\delta_{p,k''}(i))$ 
becomes zero on some node $i$; they remain zero in subsequent iterations.
Thus the algorithm terminates after $O(m)$ paths are obtained, 
where each path is found in $O(n)$ time by keeping $\{e \in \delta_{p,k}(i) \mid \zeta(e) > 0\}$ $(i \in V, k =1,2,3)$ as lists.

We estimate the range in which an optimal potential exists.
Let $\mGamma_0$ denote the minimal subtree in $\mGamma$ containing $\{ q_s\}_{s \in S}$, 
and let $d(\mGamma_0)$ denote the diameter of $\mGamma_0$, i.e., $d(\mGamma_0) := \max_{u,v \in \mGamma_0} d(u,v)$. 
\begin{Lem}\label{lem:exists}
	There is an optimal potential $(p,r)$ with
	$p(i) \in \mGamma_0$ and $r(i) \leq d(\mGamma_0)$ for $i \in V$.
\end{Lem}
\begin{proof}
	Let $(p,r)$ be a potential.
	Suppose that there is a nonterminal node $i^*$ with $p_{i^*} \not \in \mGamma_0$. 
	Take such $i^*$ having the maximum distance 
	$d(p_{i^*},\varGamma_0) := \min_{u \in \mGamma_0} d(p_{i^*}, u)$ from $\mGamma_0$.
	We can assume that $\mGamma^*_{p_{i^*},1}$ contains $\mGamma_0$.
	Let $X$ be the set of nodes $j$ with $p_j = p_{i^*}$.
	Suppose that $p_{i^*} \in \mGamma^* \setminus \mGamma$. Then $r_j \geq 1/2$ for all $j \in X$. 
	For each $j \in X$, 
	replace $(p_j, r_j)$ by $(p_{j \to^* 1}, r_j - 1/2)$.
	For an edge $ij$ with $i \in X$ and $j \not \in X$, 
	both $d(p_i, p_j)$ and $r_i + r_j$ decrease by $1/2$, and thus (p2) remains to hold.
	For other edge $ij$,  quantity
	$d(p_i,p_j) - r_i - r_j$ is nonincreasing or remains nonpositive (if $i,j \in X$). 
	The feasibility (p2) still holds (since $a(ij)$ is nonnegative).
	Thus the resulting $(p,r)$ is a potential, and the objective value decreases.
	Suppose that $p_{i^*} \in \mGamma$.
	For each $j \in X$, 
	replace $(p_j,r_j)$ by $(p_{j \to_1}, r_j)$.
	For each edge $ij$, distance $d(p_i,p_j)$ does not increase.
	Thus the feasibility (p2) holds, and the objective value does not change.
	By repeating this procedure, we can make $(p,r)$ 
	so that $p_i \in \mGamma_0$ for $i \in V$, without increasing the objective value.
	Suppose that $r_i > d(\mGamma_0)$ for some $i$; necessarily $r_i \geq 1$.
	For each edge $ij$ connecting $i$, 
	it holds $d(p_i,p_j) - r_i - r_j - a_{ij} \leq -1$ 
	(since $d(p_i,p_j) \leq d(\mGamma_0$) and $d(p_i,p_j) - r_i - r_j$ is an integer).
	Thus we can replace $r_i$ by $r_i - 1$ to decrease the objective value.
	Repeating this procedure, $(p,r)$ satisfies $r_i \leq d(\mGamma_0)$, as required.
\end{proof}

\paragraph{Reduction to a nondegenerate instance.}
Here we explain how to reduce our original problem to a nondegenerate problem.
An instance of the original problem is viewed as a pair of network $N =(V,E,S,c,a)$ 
and a tree-embedding ${\cal E} = (\mGamma, \{v_s\}_{s \in S})$ 
such that $a(e) = 0$ for all edges $e$ and $\mGamma$ is a star with center $v_0$ and leaves $v_s$ $(s \in S)$.
We are going to
construct a nondegenerate instance.
Define edge cost $\tilde a$ by $\tilde a(e) := 2$ for each edge $e \in E$.
Let $\tilde N := (V,E,S,c,\tilde a)$.
Next we define a tree-embedding $\tilde{\cal E} = (\tilde \varGamma, \{q_s\}_{s \in S})$.
Let $\mSigma$ be any (finite) trivalent tree 
with $|S|$ leaves $u_s$ $(s \in S)$ and diameter $D = O(\log |S|)$.
For each $s \in S$, 
consider an infinite path $P_s$ having a vertex $u'_s$ of degree one.
Identify $u_s$ and $u'_s$, i.e., glue $P_s$ and $\mSigma$ at $u_s$.
The resulting infinite tree is denoted by $\tilde \varGamma$.
Define $q_s$ as the vertex in $P_s$ having distance $(2|E| + 1) D$ from $u_s (= u'_s)$.
See Figure~\ref{fig:perturbation} for the construction of $\tilde \varGamma$.
\begin{figure} 
	\begin{center} 
		\includegraphics[scale=0.5]{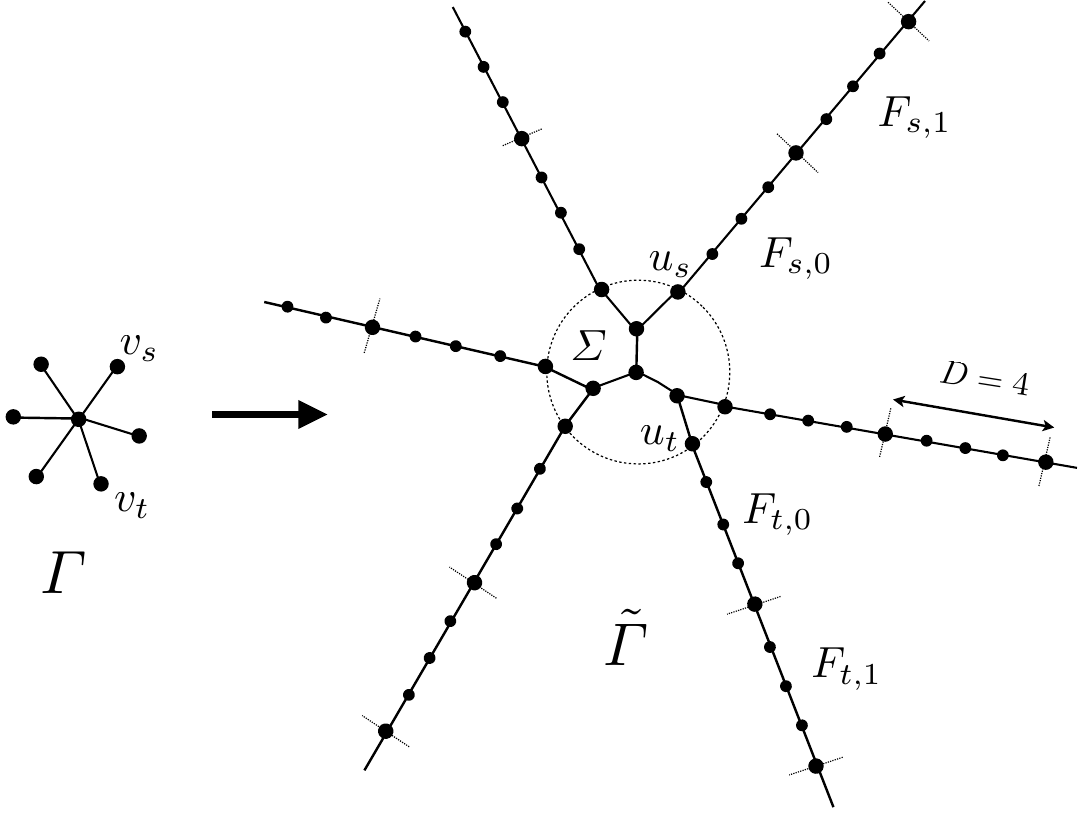}
		\caption{Construction of $\tilde \varGamma$}  
		\label{fig:perturbation}         
	\end{center}
\end{figure} 

Now we obtain a nondegenerate instance $(\tilde N, \tilde {\cal E})$.
Let $(\tilde p,\tilde r)$ and $f = ({\cal P},\lambda)$ 
be an optimal potential and an optimal multiflow, respectively, for this perturbed instance $(\tilde N, \tilde {\cal E})$.
We show that $f$ is a maximum multiflow, i.e., 
optimal for the original instance $(N, {\cal E})$.
We are going to construct an optimal potential $(p,r)$ for $(N, {\cal E})$ from $(\tilde p,\tilde r)$.
Let $B_i$ be the set of vertices $q$ with $d(\tilde p_i, q) \leq \tilde r_i$.
Namely  $B_i$ is the ball with center $\tilde p_i$ and radius $\tilde r_i$.
By (p1), vertices $q$ with $d(\tilde p_i, q) = \tilde r_i$ belong to $\tilde \varGamma$.
Hence we can identify $B_i$ with the subgraph of $\tilde \varGamma$ 
induced by $B_i \cap \tilde \varGamma$.
Then it holds
\begin{equation}
d(B_i,B_j) := \min_{u \in B_i,v \in B_j} d(u,v) = \max \{ 0, d(\tilde p_i, \tilde p_j) - \tilde r_i - \tilde r_j\}.
\end{equation}
For $k=0,1,2,\ldots,2|E|$ and $s \in S$, let $F_{s,k}$ denote the subgraph of $\tilde \varGamma$
consisting of edges $uv$ in $P_s \subseteq \tilde \varGamma$ such that
$kD \leq d(u_s,u) = d(u_s,v)- 1 < (k + 1)D$.
Let $F_k$ be the union of $F_{s,k}$ over $s \in S$.
We say that an edge $ij$ in $E$ {\em hits} $F_k$ if 
$B_i \cap B_j = \emptyset$ and the path between $B_i$ and $B_j$ meets an edge of $F_k$.
For an edge $ij$, it holds $d(B_i,B_j) \leq \tilde{a}_{ij} = 2$.
Therefore
the path between $B_i$ and $B_j$ consists of at most two edges.
Since $2|E| + 1$ subgraphs $F_0, F_1,\ldots,F_{2|E|}$ are edge-disjoint,  
there is an index $k$ such that every edge in $E$ does not hit $F_k$. 
Fix such an index $k$.

For each $s \in S$, choose the edge $e_s$ of $F_{s,k}$ nearest to $q_s$ (furthest from $u_s$).
Delete all $e_s$ from $\tilde \varGamma$.
There are $|S| + 1$ connected components $C_0, C_{s}$ $(s \in S)$, 
where $C_0$ is the connected component containing $\varSigma$,  
and $C_s$ is the connected component containing $q_s$.
For each $i \in V$, define $(p_i,r_i) \in \varGamma^* \times \ZZ/2$ by
\begin{equation}
(p_i,r_i) := \left\{ \begin{array}{ll}
(v_0, 1)  & {\rm if}\ \mbox{$B_i$ contains two $e_s,e_{s'}$,} \\
(\bar v_s, 1/2) & {\rm if}\ \mbox{$B_i$ contains exactly one $e_s$,} \\
(v_s,0) & {\rm if}\ \mbox{$B_i$ is contained in $C_{s}$ for $s \in S \cup \{0\}$,}
\end{array}\right. 
\end{equation}
where $\bar v_s$ denotes the vertex in $\mGamma^* \setminus \mGamma$ obtained by subdividing 
edge $v_0v_s$ in $\mGamma$.
We show that $(p,r)$ is a potential for $(N,{\cal E})$ and satisfies (o1), (o2), and (o3) with $f$.
We first show the feasibility (p2) $d_{\mGamma}(p_i,p_j) \leq r_i + r_j$ for $ij \in E$.
Since $p_i \in \mGamma^* \setminus \mGamma$ implies $r_i = 1/2$,
we may consider the three cases: 
(i) $p_i = v_s$, $p_j = v_0$,
(ii) $p_i = v_s$, $p_j = \bar v_{s'}$ for $s \neq s'$,
and (iii) $p_i = v_s$, $p_j = v_{s'}$ for $s \neq s'$.
%
For (i), $B_j$ necessarily contains two $e_s, e_{s'}$; otherwise $ij$ hits $F_k$ at $e_s$.
This implies $r_j = 1$.
(ii) and (iii) cannot occur since, otherwise, $ij$ hits $F_k$ at $e_s$. 
Thus (p2) holds, and hence $(p,r)$ is a potential; (p1) and (p3) are clearly satisfied.

Next we show (o1),(o2), and (o3) for $f$ and $(p,r)$.
To show (o2),
take an edge $ij$ with $f(ij) > 0$.
By (o2) for $f$ and $(\tilde p,\tilde r)$ in $(\tilde N, \tilde {\cal E})$,
it holds $d_{\tilde \varGamma}(\tilde p_i, \tilde p_j) - \tilde r_i - \tilde r_j  = 2$.
Thus the balls $B_i$ and $B_j$ are disjoint and have distance $2$.
Suppose that $B_i$ has two edges $e_s, e_{s'}$.
Then $B_i$ must meet $F_{s,k}$ for every $s \in S$.
Necessarily $B_j$ cannot be contained by $C_0$.
Also $B_j$ cannot have $e_t$ for any $t \in S \setminus \{s,s'\}$; otherwise $ij$ hits $F_k$.
Thus $B_j$ is contained in $C_t$ for $t \in S$, and
$(p_i,r_i) = (v_0,1)$ and $(p_j,r_j) = (v_t, 0)$ hold.
If $B_i$ has (only one) $e_{s}$ and $B_j$ has (only one) $e_{s'}$, then 
$s$ and $s'$ must be different, and necessarily $(p_i,r_i) = (\bar v_s,1/2)$ 
and $(p_j,r_j) = (\bar v_{s'}, 1/2)$. 
If $B_i$ has only $e_s$ and $B_j$ does not have any of $e_t$, 
then necessarily $B_j$ is contained in $C_s$ or $C_0$; 
hence $(p_i,r_i) = (\bar v_s, 1/2)$ and $(p_j,r_j) = (v_s,0)$ or $(v_0,0)$.
If both $B_i$ and $B_j$ do not have any of $e_s$, 
then both $B_i$ and $B_j$ are contained in $C_s$ for some $s \in S \cup \{0\}$, 
and $p_i = p_j$ and $r_i = r_j = 0$.
In all the cases, it holds $d_{\mGamma}(p_i,p_j) = r_i + r_j$, 
implying (o2).

Consider the condition (o1).
Take a path $P = (s = j_0,j_1,\ldots,j_m = t)$ with $\lambda(P) > 0$.
There is an index $l$ such that $B_{l}$ contains $e_s$; 
otherwise $F_k$ is hit by some edge.
Moreover such an index $l$ is unique.
Otherwise, the balls $B_{j_l}$ and $B_{j_{l'}}$ with $l < l'$ contain $e_s$.
Then $d(\tilde p_{j_l}, \tilde p_{j_{l'}}) - \tilde r_{j_l} - \tilde r_{j_{l'}} < 0$.
However, by (o1) and (o2) for $f$ and $(\tilde p, \tilde r)$, we have
$d(\tilde p_{j_l},\tilde p_{j_{l'}})  = \sum_{i = l}^{l'-1} d(\tilde p_{j_i},\tilde p_{j_{i+1}}) = 
\sum_{i= l}^{l'-1} (\tilde r_{j_i} + \tilde r_{j_{i+1}} + 2)$, 
and $d(\tilde p_{j_l},\tilde p_{j_{l'}})   
  - \tilde r_{j_l} - \tilde r_{j_{l'}} \geq 2 > 0$; this is a contradiction.
Similarly there is a unique index $l'$ such that $B_{l'}$ contains $e_t$.
If $l = l'$, then $(p_{j_0},p_{j_1},\ldots,p_{j_m})$ must be $(v_s,v_s,\ldots,v_s,v_0,v_t,\ldots,v_t)$.
If $l \neq l'$, say $l < l'$, then $(p_{j_0},p_{j_1},\ldots,p_{j_m})$ must be
$(v_s,\ldots,v_s,\bar v_s,v_0,\ldots,v_0,\bar v_t,v_t,\ldots,v_t)$ or $(v_s,\ldots,v_s,\bar v_s,\bar v_t,v_t,\ldots,v_t)$.
Thus we obtain (o1).
Since $\tilde r(i) = 0$ implies $r(i) = 0$, we obtain (o3).
Hence $(p,r)$ is an optimal potential, and $f$ is a maximum multiflow.

\subsection{Discrete convexity and steepest descent algorithm (SDA)}\label{subsec:dc}
Here we briefly introduce a class of discrete convex functions (L-convex functions) 
on a certain graph structure 
and the steepest descent algorithm to minimize them.
We then
explain that our problem falls into the minimization of an L-convex function.
A general theory is given in~\cite{HHprepar}; see also \cite{HH17survey}.

First we equip the space of all potentials with a graph structure.
Let $\ZZ^* (:= \ZZ/2)$ denote the set of half-integers.
Let $\mGamma^* \boxtimes \ZZ^*$ 
denote the set of pairs $(p,r) \in \mGamma^* \times \ZZ^*$
such that $p \in \mGamma$ if and only if $r \in \ZZ$. 
Two points $(p,r)$ and $(p',r')$ are adjacent if and only 
if $p$ and $p'$ are adjacent in $\mGamma^*$  
and $|r - r'| = 1/2$.
Fix an arbitrary vertex $p_0$ of $\mGamma$.
Let $B$ (resp.  $W$) denote the subset of $\mGamma^* \boxtimes \ZZ^*$  
consisting of pairs $(p,r) \in \mGamma \times \ZZ$ with
$d(p,p_0) + r$ even (resp. odd).
Orient each edge of $\mGamma^* \boxtimes \ZZ^*$ 
by $(p,r) \leftarrow (p',r')$ if $(p,r) \in B$ or $(p',r') \in W$.
Namely $B$ is the set of sinks and $W$ is the set of sources.
This orientation is acyclic, 
and induces a partial order $\preceq$ on $\mGamma^* \boxtimes \ZZ^*$.
\begin{figure} 
	\begin{center} 
		\includegraphics[scale=0.6]{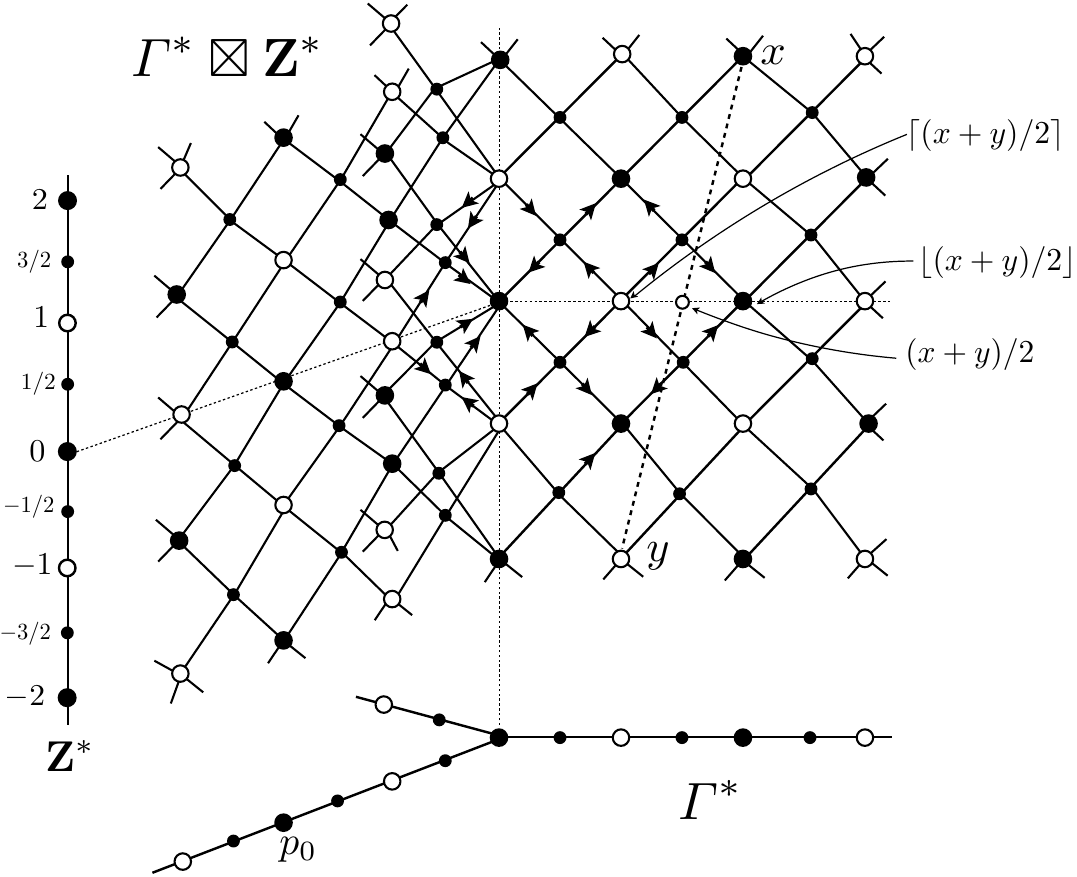}
		\caption{Graph $\mGamma^* \boxtimes \ZZ^*$}  
		\label{fig:building}         
	\end{center}
\end{figure} 
See Figure~\ref{fig:building}, where 
nodes in $B$ and $W$ are colored black and white, respectively.

Next we define midpoint operations on $\mGamma^* \boxtimes \ZZ^*$.
Let $\mGamma^{**}$ denote the edge-subdivision of $\mGamma^*$ with edge-length $1/4$, 
let $\ZZ^{**} (:= \ZZ/4)$ denote the set of quarter-integers, 
and let $\mGamma^{**} \boxtimes \ZZ^{**}$ 
denote the set of pairs $(p,r) \in \mGamma^{**} \times \ZZ^{**}$
such that $p \in \mGamma^*$ if and only if $r \in \ZZ^*$. 
For two points $x = (p,r), x'=(p',r')$ in $\mGamma^* \boxtimes \ZZ^*$, 
there exists a unique {\em midpoint} 
$y = (q, t)\in \mGamma^{**} \boxtimes \ZZ^{**}$ such that
$d(p,q) + d(q,p') = d(p,p')$, $d(p,q) = d(q,p')$, and $t = (r+r')/2$.
This $y$ is denoted by $(x+x')/2$;  accordingly $q$ is denoted by $(p+p')/2$.
For $z = (q,t) \in \mGamma^{**} \boxtimes \ZZ^{**}$, 
there uniquely exists a pair $(x,y)$ of vertices in  $\mGamma^{*} \boxtimes \ZZ^{*}$ 
with the property that $z = (x+y)/2$ and $x \preceq y$.
We denote $x$ and $y$ by $\lfloor z \rfloor$ and $\lceil z \rceil$, respectively.

We are ready to define L-convex functions. 
For a natural number $n$, consider the product $(\mGamma^* \boxtimes \ZZ^*)^n$; 
a point $x$ in $(\mGamma^* \boxtimes \ZZ^*)^n$ is represented 
by a pair $(p,r)$ of $p \in (\mGamma^*)^n$ and $r \in (\ZZ^*)^n$.
A function $g: (\mGamma^* \boxtimes \ZZ^*)^n \to \RR \cup \{\infty\}$
is {\em L-convex} if it satisfies the following analogue of the {\em discrete midpoint convexity}~\cite[Section 7.2]{MurotaBook}:
\begin{equation}\label{eqn:midpoint}
g(x) + g(y) \geq g\left( \lfloor (x+y)/2 \rfloor \right) + g \left( \lceil (x+y)/2 \rceil \right) \quad (x,y \in  (\mGamma^* \boxtimes \ZZ^*)^n),
\end{equation}
where $ (\lfloor (x+y)/2 \rfloor)_i :=  \lfloor (x_i + y_i)/2] \rfloor$ and 
$(\lceil (x + y)/2 \rceil)_i := \lceil (x_i + y_i)/2 \rceil$ for $i=1,2,\ldots,n$.

For $x \in (\mGamma^* \boxtimes \ZZ^*)^n$, 
let ${\cal F}_x$ (resp. ${\cal I}_x$) denote the set of points $y$ with $x_i \preceq y_i$ (resp. $x_i \succeq y_i$) for $i=1,2,\ldots,n$.
The set ${\cal F}_{x} \cup {\cal I}_{x}$ is called  the {\em neighborhood} of $x$.
The {\em steepest descent algorithm} is given as follows.
\begin{description}
	\item[Algorithm 2:] Steepest descent algorithm (SDA)
	\item[Input:] 
	An L-convex function $g: (\mGamma^* \boxtimes \ZZ^*)^n \to \RR \cup \{\infty\}$, 
	and a point $x^0$ with  $g(x^0) < \infty$.
	\item[Output:] A minimizer of $g$.
	\item[Step 0:] Let $i \leftarrow 0$.
	\item[Step 1:] Find a minimizer $y$ of $g$ over the neighborhood ${\cal F}_{x^i} \cup {\cal I}_{x^i}$ of $x^i$.
	\item[Step 2:] If $g(x^i) = g(y)$, then output $x^i$ and stop; $x^i$ is a minimizer.
	\item[Step 3:] Otherwise, let $x^{i+1} \leftarrow y$,  $i \leftarrow i+1$, and go to step 1.
\end{description}
The fact that the output is a minimizer easily follows from (\ref{eqn:midpoint}); 
see \cite[Theorem 2.5]{HH14extendable}.

We discuss the number of iterations of this algorithm.
For $x,y \in (\mGamma^* \boxtimes \ZZ^*)^n$, 
an {\em $l_{\infty}$-path} between $x$ and $y$ 
is a sequence $P = (x = x^0,x^1,\ldots,x^m = y)$ such that for each $k$ and $i$,
the $i$-th components $x^k_i$ and $x^{k+1}_i$ 
belong to a 4-cycle of $\mGamma^* \boxtimes \ZZ^*$, 
i.e., it hold 
$z \preceq x_i^k \preceq z'$ and $z \preceq x_i^{k+1} \preceq z'$ 
for some $z \in B, z' \in W$ with $z \preceq z'$. 
The {\em length} of $P$ is defined as $m$.
The {\em $l_{\infty}$-distance} between $x$ and $y$, denoted by $D_{\infty}(x,y)$, 
is defined as the minimum length of an $l_{\infty}$-path between $x$ and $y$. 
For distinct $(p,r), (q,s) \in \mGamma^* \boxtimes \ZZ^*$ in a 4-cycle,
it holds $d(p ,q) + |r - s| = 1$. From this we observe that
\begin{equation*}
D_{\infty}(x,y) = \max_{1 \leq i \leq n} (d(p_i,q_i) + |r_i - s_i|)  \quad 
(x = (p,r),y = (q,s) \in (\mGamma^* \boxtimes \ZZ^*)^n).
\end{equation*}
%
Notice that a sequence $(x = x^0,x^1,x^2,\ldots,x^{m})$ generated 
by SDA is an $l_{\infty}$-path.
Let $\opt (g)$ denote the set of all minimizers of $g$.
The length $m$, which is the number of the iterations, 
is at least $D_{\infty}(x_0, \opt(g)) = \min_{y \in \opt(g)} D_{\infty}(x_0, y)$. 
This lower bound is almost tight.
\begin{Thm}[\cite{HHprepar}]\label{thm:bound}
	Let $(x = x^0, x^1, \ldots, x^m)$ be a sequence of points generated by SDA 
	applied to an L-convex function $g$ and an initial point $x$. 
	Then $m \leq D_{\infty}(x, \opt(g)) + 2$.
	If $g(x) = \min_{y \in {\cal F}_x} g(y)$ or $g(x) = \min_{y \in {\cal I}_x} g(y)$, 
	then $m = D_{\infty}(x, \opt(g))$.
\end{Thm}
A similar bound 
for original L-convex functions in DCA was established in~\cite{MurotaShioura14}. 
For a similar but different class of L-convex functions 
(called {\em alternating L-convex functions}), 
the same bound was proved by~\cite{HH14extendable}.
By using this result, we give a shorter proof of Theorem~\ref{thm:bound} as follows.
\begin{proof}
First consider the case where $\mGamma$ is an infinite path (without ends).
Then $\mGamma^* \boxtimes \ZZ^*$ is isomorphic to the product of two zigzagly-oriented paths.
In this case, $(\mGamma^* \boxtimes \ZZ^*)^{n}$ 
is identified with the product  $(\ZZ^*)^{2n}$ of $2n$ paths, and
L-convex functions on $(\mGamma^* \boxtimes \ZZ^*)^{n}$ 
coincide with alternating L-convex functions in the sense of~\cite{HH14extendable}.
Also $D_{\infty}$ is equal to $d$ in~\cite{HH14extendable}.
Then Theorem~\ref{thm:bound} was shown in \cite[Theorem 2.6]{HH14extendable}. 
In particular, the following holds:
\begin{itemize}
	\item[($*$)] For $x \in (\mGamma^* \boxtimes \ZZ^*)^{n}$ with $g(x) < \infty$
	and a minimizer $x'$ of $g$ over ${\cal F}_x \cup {\cal I}_{x}$ with $g(x') < g(x)$, 
	if $g(x) = \min_{y \in {\cal I}_{x}} g(y)$ or $g(x) = \min_{y \in {\cal F}_{x}} g(y)$, 
	then it holds 
	\begin{equation}\label{eqn:D_inf=D_inf-1}
		D_{\infty}(x', \opt(g)) = D_{\infty}(x, \opt(g)) - 1.
	\end{equation}
\end{itemize}
We show that this proposition holds for a general tree $\mGamma$.
Pick  an arbitrary $z \in \opt(g)$ with $D_{\infty}(x, \opt(g)) = D_{\infty}(x, z)$. 
Since $(x'_i)_p$ and $(x_i)_p$ are equal or adjacent in $\mGamma$ or $\mGamma^*$, 
there is a path $P_i$ in $\mGamma^*$ containing $(x_i)_p, (x'_i)_p$, and $(z_i)_p$.
Then $z,x$, and $x'$ are points in $\prod_{i=1}^n P_i \boxtimes \ZZ^* \simeq (\ZZ^*)^{2n}$, 
and the restriction $g'$ of $g$ to $\prod_{i=1}^n P_i \boxtimes \ZZ^*$
is (alternating) L-convex.
Thus ($*$) is applicable to $g', x, x'$, and hence 
$D_{\infty}(x', \opt(g')) = D_{\infty}(x, \opt(g')) - 1 = D_{\infty}(x, \opt(g)) - 1$.
By $D_{\infty}(x', \opt(g)) \leq D_{\infty}(x', \opt(g'))$, 
the equation (\ref{eqn:D_inf=D_inf-1}) holds.

Theorem~\ref{thm:bound} is proved as follows.
By the description of the steepest descent algorithm,
it holds that $g(x^i) = \min_{y \in {\cal I}_{x^i}} g(y)$ or $g(x^i) = \min_{y \in {\cal F}_{x^i}} g(y)$
for all $i > 0$.
Thus $m - 1  = D_{\infty}(x^1, \opt(g))$ holds.
Since $D_{\infty}(x^1, \opt(g)) \leq  D_{\infty}(x, \opt(g)) + 1$, 
we obtain $m \leq D_{\infty}(x, \opt(g)) + 2$. 
In addition, if $g(x) = \min_{y \in {\cal I}_{x}} g(y)$ or $g(x) = \min_{y \in {\cal F}_{x}} g(y)$, 
then $m  = D_{\infty}(x, \opt(g))$ holds.
\end{proof}

We now return to our problem.
Let a pair of $N = (V,E,S,c,a)$ and ${\cal E} = (\mGamma, \{q_s\}_{s \in S})$ be an instance.
Let $V = \{1,2,\ldots,n\}$.
Then the set of  potentials $(p,r) \in (\mGamma^*)^V \times (\ZZ_{+}/2)^V$ is naturally regarded 
as a subset of $(\mGamma^* \boxtimes \ZZ^*)^n$. 
Define a function $g_{N, {\cal E}}: (\mGamma^* \boxtimes \ZZ^*)^n \to \RR \cup \{\infty\}$
by
\begin{equation*}
g_{N, {\cal E}}(x) := \left\{
\begin{array}{ll}
\sum_{i \in V \setminus S}2 c(i) r(i) & {\rm if}\ \mbox{$x = (p,r)$ is a potential}, \\
\infty & {\rm otherwise}.
\end{array} \right. 
\end{equation*}
\begin{Prop}\label{prop:g_is_L-convex}
Suppose that $a$ is even-valued.	
Then $g_{N, {\cal E}}$ is L-convex.
\end{Prop}
By this proposition, the dual of our multiflow problem can be viewed 
as an L-convex function minimization. 
Therefore the optimality check in SDA (steps 1 and 2) 
must be equivalent to that (Lemmas~\ref{lem:optimality'} and \ref{lem:mn}) obtained from the multiflow duality in Section~\ref{sec:multiflow}; 
see also Figure~\ref{fig:outline}. 

The rest of this section is devoted to proving Proposition~\ref{prop:g_is_L-convex}.
For $x \in \mGamma^{**} \boxtimes \ZZ^{**}$, the first and second components of $x$ are denoted by $x_p$ and $x_r$, respectively, i.e., $x = (x_p, x_r)$.
For a nonnegative even integer $a \geq 0$, 
define $h  = h_a: (\mGamma^{**} \boxtimes \ZZ^{**})^2 \to \RR$ by
\[
h(x,y) := d(x_p, y_p) - x_r - y_r - a \quad (x,y \in \mGamma^{**} \boxtimes \ZZ^{**}).
\]
Notice that if
$x=(p,r) \in (\mGamma^{*} \boxtimes \ZZ^{*})^n$ is a potential
then $h_{a_{ij}}(x_i, x_j) \leq 0$ for $ij \in E$.
\begin{Lem}\label{lem:h_is_convex}
	For $x,y,x',y' \in \mGamma^{*} \boxtimes \ZZ^{*}$, it holds
	$
	h(x,y) + h(x',y') \geq 2 h((x+x')/2, (y+y')/2).
	$
\end{Lem}
\begin{proof}
	A classical result~\cite[Lemma 3]{DFL76} in location theory says that 
	the distance function on a tree is convex. 
	Thus it holds 
	$
	d(p, q) + d(p', q') \geq 2 d((p+p')/2, (q+ q')/2)
	$ for $p,q,p',q' \in \mGamma^*$.
	The inequality immediately follows from this fact. 
\end{proof}

\begin{Lem}\label{lem:h(x,y)<=0}
	For $x,y \in \mGamma^{**} \boxtimes \ZZ^{**}$ with $x_r,y_r \geq 0$, 
	if $h(x,y) \leq 0$, then $h(\lceil x \rceil, \lceil y \rceil) \leq 0$ 
	and $h(\lfloor x \rfloor, \lfloor y \rfloor) \leq 0$.
\end{Lem}
\begin{proof}
	Let $\mDelta := h(\lceil x \rceil, \lceil y \rceil) - h(x,y)$ 
	and $\mDelta' := h(\lfloor x \rfloor, \lfloor y \rfloor) - h(x,y)$.
	From $h(\lceil x \rceil,y) - h(x,y) \in \{-1/2,0,1/2\}$, 
	we see that $\mDelta, \mDelta' \in \{-1, -1/2, 0, 1/2, 1\}$.
	Notice that both $h(\lceil x \rceil, \lceil y \rceil)$ and $h(\lfloor x \rfloor, \lfloor y \rfloor)$ are integers.
	Therefore, $h(x,y) \leq -1/2$ implies 
	$h(\lceil x \rceil, \lceil y \rceil) \leq 0$
	and $h(\lfloor x \rfloor, \lfloor y \rfloor) \leq 0$.
	Hence we assume that $h(x,y) = 0$ and 
	$\mDelta, \mDelta' \in \{- 1, 0, 1\}$.

	Case 1: $d(x_p, y_p) < 1/2$, i.e., $d(x_p, y_p) = 0$ or $1/4$.
	Suppose $d(x_p,y_p) = 0$.  
	By $a, x_r,y_r \geq 0$ and $h(x,y) = 0$ 
	we have $a = 0$ and $x_r = y_r = 0$. 
	Then $\lceil x \rceil_p = \lceil y \rceil_p$,
	$\lfloor x \rfloor_p = \lfloor y \rfloor_p$ and 
	$\lceil x \rceil_r = \lceil y \rceil_r = \lfloor x \rfloor_r = \lfloor y \rfloor_r = 0$.
	Suppose $d(x_p, y_p) = 1/4$.
	We can assume $(d(x_p, y_p), x_r,y_r) = (1/4,0,1/4)$.
	Then  $(d(\lceil x \rceil_p,\lceil y \rceil_p), \lceil x \rceil_r, \lceil y \rceil_r)$
	and
	$(d(\lfloor x \rfloor_p, \lfloor y \rfloor_p),\lfloor x \rfloor_r, \lfloor y \rfloor_r)$
	are $(0,0,0)$ or $(1/2,0,1/2)$. 
	In both cases, $h(\lceil x \rceil, \lceil y \rceil) = h(\lfloor x \rfloor,\lfloor y \rfloor) = 0$ holds.
	
	Case 2: $d(x_p,y_p) \geq 1/2$.
	In this case, the simple paths in $\varGamma^{**}$ 
	connecting $\lceil x \rceil_p, \lfloor x \rfloor_p$ and 
	$\lceil y \rceil_p, \lfloor y \rfloor_p$, respectively, are edge-disjoint, 
	and is contained in a single path in $\varGamma^{**}$. 
	From this, we observe that
	$
	d(\lceil x \rceil_p,\lceil y \rceil_p) - d(x_p,y_p) = d(x_p,y_p) - d(\lfloor x \rfloor_p, \lfloor y \rfloor_p)$.
	By $\lceil x \rceil_r - x_r = x_r - \lfloor x \rfloor_r$ and 
	$\lceil y \rceil_r - y_r = y_r - \lfloor y \rfloor_r$, we obtain
	$\mDelta' = - \mDelta$. 
	So it suffices to show that $\mDelta = 1$ cannot occur.
	Suppose not. Then both $h(\lceil x \rceil, \lceil y \rceil)$ and $h(\lfloor x \rfloor, \lfloor y \rfloor)$ are odd.
	Observe that two colored nodes $z,z' \in \varGamma^* \boxtimes \ZZ^*$ 
	have the same color if and only if integer $d(z_p,z'_p) + z_r + z_{r}'$ is even
	(see Lemma~\ref{lem:parity}).
	If both $\lceil x \rceil$ and $\lceil y \rceil$ (or $\lfloor x \rfloor$ and $\lfloor y \rfloor$) are colored, 
	then 
	they must have different colors (by the evenness of $a$),
	$x = \lceil x \rceil = \lfloor x \rfloor$ or $y = \lceil y \rceil = \lfloor y \rfloor$ holds, 
	and $\mDelta = 1$ is impossible.
	Hence we can assume that $x_p, y_p \in \mGamma^{**} \setminus \mGamma^*$, and
	$(\lceil x \rceil, \lfloor y \rfloor) \in W \times B$ or $(\lfloor x \rfloor, \lceil y \rceil) \in B \times W$.
	With $\varDelta = - \varDelta' = 1$, it must hold
	$d(\lceil x \rceil_p, \lceil y \rceil_p) - d(x_p, y_p) = d(x_p,y_p) - d(\lfloor x \rfloor_p, \lfloor y \rfloor_p) = 1/2$ and $\lceil x \rceil_r - x_r  = \lceil y \rceil_r - y_r  = x_r - \lfloor x \rfloor_r = y_r - \lfloor y \rfloor_r = -1/4$.
	Also 
	$d(\lceil x \rceil_p,  \lfloor y \rfloor_p ) = d(\lfloor x \rfloor_p, \lceil y \rceil_p) =d(x_p, y_p)$.
	Hence $h(\lceil x \rceil, \lfloor y \rfloor) = h(\lfloor x \rfloor, \lceil y \rceil) = h(x,y) = 0$ (even).
	However
	$(\lceil x \rceil, \lfloor y \rfloor) \in W \times B$ or $(\lfloor x \rfloor, \lceil y \rceil) \in B \times W$ implies
	that $h(\lceil x \rceil, \lfloor y \rfloor)$ or $h(\lfloor x \rfloor, \lceil y \rceil)$ is odd; a contradiction.  
\end{proof}

\begin{proof}[Proof of Proposition~\ref{prop:g_is_L-convex}]
It is easy to see that the (linear) function $(p,r) \mapsto \sum_{i \in V \setminus S}2 c(i) r(i)$ is L-convex.
Since L-convexity is preserved under nonnegative combination, 
it suffices to show that the function $g$ on $(\mGamma^* \boxtimes \ZZ^*)^2$ 
defined by $(x,y) \mapsto 0$ if $x_r,y_r \geq 0$, $h(x,y) \leq 0$, and~$\infty$ otherwise
is L-convex.
For $z \in \varGamma^{**} \boxtimes \ZZ^{**}$ with $z_r \geq 0$, 
it is easy to see $\lceil z \rceil_r, \lfloor z \rfloor_r \geq 0$.
Consider $(x,y),(x',y') \in \varGamma^* \boxtimes \ZZ^*$ 
with $x_r,y_r,x'_r,y'_r \geq 0$.
Suppose that $h(x,y) \leq 0$ and $h(x',y') \leq 0$.
By Lemma~\ref{lem:h_is_convex}, we have $h((x+x')/2, (y+y')/2) \leq 0$.
By Lemma~\ref{lem:h(x,y)<=0}, we have
$h( \lceil (x+x')/2 \rceil, \lceil (y+y')/2 \rceil ) \leq 0$ and 
$h( \lfloor (x+x')/2 \rfloor, \lfloor (y+y')/2 \rfloor) \leq 0$. 
Thus $g$ is L-convex.
\end{proof}

\section{Algorithm}\label{sec:algo}
In this section, we deal with the lower right part of Figure~\ref{fig:outline}
and present an algorithm ({\em dual descent algorithm}) to solve 
a nondegenerate instance of our multiflow problem, 
with proving of the main result (Theorem~\ref{thm:main}).
We first show that the optimality check of a potential $(p,r)$, or finding a $(p,r)$-admissible support, 
is reduced to the submodular flow feasibility 
problem on a certain skew-symmetric network, 
called the {\em double covering network}.
This extends the earlier result on the minimum-cost edge-capacitated multiflow problem by Karzanov~\cite{Kar79, Kar94},  
in which the optimality is checked by the classical circulation problem.
Partial adaptations of this idea to the node-capacitated setting  
have been given in \cite{BK08ESA, BK12}; 
but the full adaptation using submodular flow is new. 
Checking the feasibility of a submodular flow is reduced 
to the maximum submodular flow problem.
We prove that the minimal minimum cut naturally gives a steepest direction at each potential.
Then we obtain a simple descent algorithm mentioned in Introduction.

\subsection{Double covering network with submodular constraints}\label{subsec:doublecovering}
Let a pair of $N = (V,E,S,c,a)$  and ${\cal E} = (\mGamma, \{q_s\}_{s \in S})$ be a nondegenerate instance.
Let $(p,r)$ be a potential for $(N, {\cal E})$. 
We construct a skew-symmetric network ${\cal D}_{p,r}$ together with submodular constraints as follows.

We first define a signed node set $U_i$ and 
directed edge set $A_i$ indexed by each node $i$ in $N$, 
together with lower edge-capacity $\underline{c}$ and upper edge-capacity $\overline{c}$.
A nonterminal node $i$ is said to be {\em flat} if $p(i) \not \in \mGamma_3$, 
and {\em singular} if $p(i) \in \mGamma_3$. 
A nonterminal node $i$ is said to be {\em zero} 
if $r(i) = 0$ and {\em positive} if $r(i) > 0$.
\begin{itemize}
\item
For each terminal $s$,  let $U_s := \{ s^+,s^-\}$, and let $A_s := \{ s^+s^-\}$ 
with $\underline{c}(s^+s^-) := 0$ and $\overline{c}(s^+s^-) := \infty$.
\item
For each flat node $i$, 
let  $U_i := \{ i_1^+,i^+_2,i^-_1,i^-_2 \}$, and let $A_i := \{i_1^+i_2^-, i_2^+i_1^-\}$ with 
$\overline{c}(i_1^+i_2^-) = \overline{c}(i_2^+i_1^-) := c(i)$.
If $i$ is positive, then 
$\underline{c}(i_1^+i_2^-) = \underline{c}(i_2^+i_1^-) := c(i)$.
Otherwise $\underline{c}(i_1^+i_2^-) = \underline{c}(i_2^+i_1^-) := 0$.
\item
For each positive singular node~$i$, 
let $U_i := \{i_0^+, i_1^+,i^+_2,i^+_3,i_0^-, i^-_1, i^-_2,i^-_3\}$ and 
let $A_i$ consist of $i^+_0i^-_0$, $i_k^+i_0^+$, $i_0^-i_k^-$ $(k=1,2,3)$ 
with $\underline{c}(i_0^+i_0^-) = \overline{c}(i_0^+i_0^-) := 2c(i)$,
$\underline{c}(i_k^+i_0^+) = \underline{c}(i_0^-i_k^-) := 0$, 
and $\overline{c}(i_k^+i_0^+) = \overline{c}(i_0^-i_k^-) := c(i)$ $(k=1,2,3)$.
\item
For each zero singular node $i$, 
let $U_i := \{ i_1^+,i^+_2,i^+_3, i^-_1, i^-_2,i^-_3\}$, and let $A_i := \emptyset$.
\end{itemize}
Next, double each (undirected) edge in $E_{p,r}$ to two directed edges as follows.
\begin{itemize}
\item
For each edge $ij \in E_{p,r}$, where $ij \in \delta_{p,k}(i)$ 
and $ij \in \delta_{p,k'}(j)$, 
consider edges $i^-_k j^+_{k'}$ and $j^-_{k'} i^+_{k}$. 
If $j$ is a terminal $s$, 
consider $i^-_k s^+$ and $s^- i_k^+$.
The lower capacity $\underline{c}$ and upper capacity~$\overline{c}$ of these edges
are defined as $0$ and $\infty$, respectively.
\end{itemize}
To represent (a2) for zero singular nodes,  
we introduce submodular constraints, where we recall the definition of $\mDelta^*_{c(i)}$ in Section~\ref{subsec:ext}.
\begin{itemize}
	\item For each zero singular node $i$, 
	consider a submodular function $\mDelta^*_{i}$ on $U_i$
	defined by $\mDelta^*_{i}(X) := \mDelta^*_{c(i)} (X')$, 
	where $X'$ is obtained from $X$ by replacing $i_k^+$ and $i_k^-$ with $k^+$ and $k^-$, respectively.
\end{itemize}
Now the {\em double covering network ${\cal D}_{p,r} = (U, A, \underline{c}, \overline{c})$ with submodular constraints} is 
defined by the disjoint union of all these (directed) edges and node sets $U_i$,  
together with submodular functions $\mDelta^*_{i}$ 
on $U_i$ for zero singular nodes $i$. 
See Figure~\ref{fig:doublecovering}, where $i$, $j$, $i'$, and $j'$ 
are positive singular, zero singular, zero flat, and positive flat nodes, respectively.
\begin{figure} 
	\begin{center} 
		\includegraphics[scale=0.7]{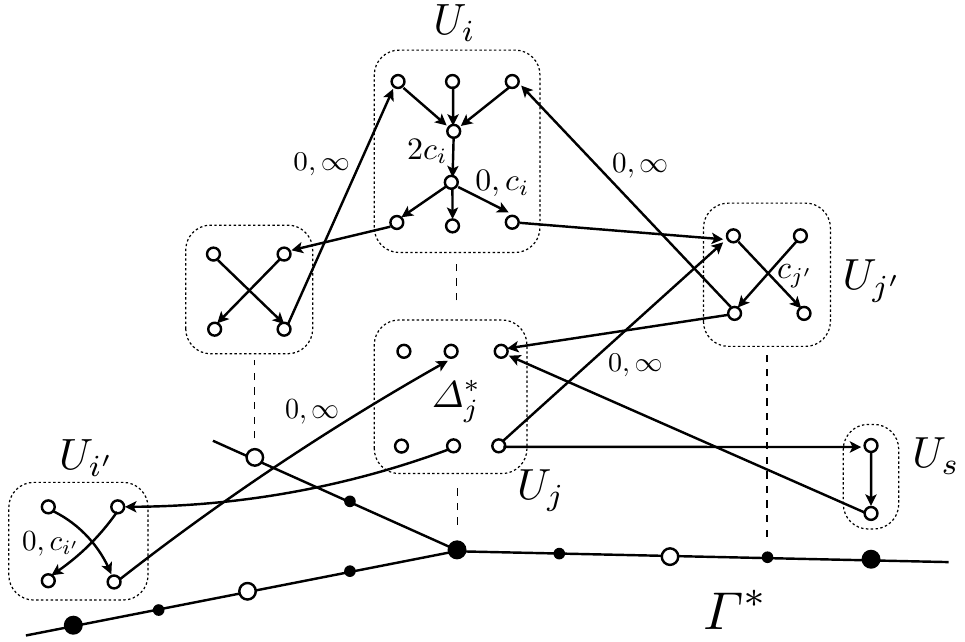}
		\caption{Double covering network}  
		\label{fig:doublecovering}         
	\end{center}
\end{figure}

%

A {\em circulation} $\varphi$ in ${\cal D}_{p,r}$
is a function on edge set $A$ such that
\begin{eqnarray*}
&& \underline{c}(e) \leq \varphi(e) \leq \overline{c}(e) \quad (e \in A), \\
&& (\nabla \varphi)(v) = 0  \quad (\mbox{$v \in U_i$ for flat or positive node $i$ in $N$}  ) ,  \\
&& (\nabla \varphi)|_{U_i} \in  {\cal B}(\mDelta^*_{i}) \quad 
(\mbox{zero singular node $i$ in $N$}).   
\end{eqnarray*}
For a circulation $\varphi$, define $\zeta_{\varphi}: E_{p,r} \to \RR_+$ by
\begin{equation}\label{eqn:zeta_phi}
\zeta_{\varphi}(e) := \frac{\varphi(e^+) + \varphi(e^-)}{2} \quad 
(e \in E_{p,r}),
\end{equation}
where $e$ is doubled to $e^+$ and $e^-$ in ${\cal D}_{p,r}$. 
\begin{Lem}
	Suppose that an integral circulation $\varphi$ in ${\cal D}_{p,r}$ exists.
	Then $\zeta_{\varphi}$ is a $(p,r)$-admissible support.
\end{Lem}
\begin{proof}
	The half-integrality of $\zeta_{\varphi}$ (in (a4)) is clear.
	Take a nonterminal node $i$.
	Let $\zeta_k  := \zeta_{\varphi}(\delta_{p,k}(i)) 
	= \sum_{e \in \delta_{p,k}(i)} (\varphi(e^+) + \varphi(e^-))/2$.
	Suppose that $i$ is flat.
	Then $\zeta_1
	= \sum_{e \in \delta_{p,1}(i)}(\varphi(e^+) + \varphi(e^-))/2 = 
	(\varphi(i_{1}^+i_{2}^-) + \varphi(i_{2}^+i_1^-))/2 = 
	\sum_{e \in \delta_{p,2}(i)}(\varphi(e^+) + \varphi(e^-))/2 = \zeta_2$.
	Since $\varphi(i_{1}^+i_{2}^-) \leq c(i)$ and $\varphi(i_{2}^+i_1^-) \leq c(i)$, 
	we obtain (a1). In addition, if $i$ is positive, then
	$\varphi(i_{1}^+i_{2}^-) = \varphi(i_{2}^+i_1^-) = c(i)$ must hold, 
	and we obtain (a3).
	Since $\zeta_1$ and $\zeta_2$ are half-integers with $\zeta_1 = \zeta_2$, 
	we have~(a4).
	
	Suppose that $i$ is positive singular.
	Since 
	$\zeta_{k} 
	= \sum_{e \in \delta_{p,k}(i)}(\varphi(e^+) + \varphi(e^-))/2
	= (\varphi(i_k^+i_0^+) + \varphi(i_0^- i_k^-))/2 \leq c(i)$, 
	we obtain $\zeta_k \leq c(i)$.
	Moreover it holds $\zeta_{\varphi}(\delta \{i\}) =  \zeta_1 + \zeta_2 + \zeta_3 = \varphi(i_0^+i_0^-) = 2 c(i)$, implying (a3) and (a4). 
	If $\zeta_1 > \zeta_2 + \zeta_3$, 
	then $\zeta_1 + \zeta_2 + \zeta_3 < 2 \zeta_1 \leq 2 c(i)$;
	this contradicts  $\zeta_1 + \zeta_2 + \zeta_3 = 2 c(i)$.
	Therefore $\zeta_1$, $\zeta_2$, and $\zeta_3$ satisfy (a2).

	Suppose that $i$ is zero singular.
	Notice that $(\nabla \varphi)(i_k^+) = \sum_{e \in \delta_{p,k}(i)} \varphi(e^+)$, and
	$(\nabla \varphi)(i_k^-) = - \sum_{e \in \delta_{p,k}(i)} \varphi(e^-)$, 
	where $e^+$ (resp. $e^-$) for $e \in \delta_{p,r}(i)$ is 
	the directed edge   
	entering (resp. leaving) $U_i$.
	Thus $(\zeta_1,\zeta_2,\zeta_3) = \phi((\nabla \varphi) |_{U_i})$; 
	see (\ref{eqn:def_phi}) for $\phi$.
	Since $(\nabla \varphi) |_{U_i} \in {\cal B}(\mDelta^*_{i})$, 
	because of  Lemmas~\ref{lem:projection} and \ref{lem:Delta*},
	the vector $(\zeta_1,\zeta_2,\zeta_3)$ satisfies (a2). 
	Since $0 = (\nabla \varphi)(U_i) = \sum_{k=1,2,3}(\nabla \varphi)(i_k^+) + \sum_{k=1,2,3}(\nabla \varphi)(i_k^-)$, we have 
	$\sum_{k=1,2,3}(\nabla \varphi)(i_k^+) = - \sum_{k=1,2,3}(\nabla \varphi)(i_k^-)$ and obtain (a4) by
	\[
	\zeta_{\varphi}(\delta \{i\}) = 
	\zeta_1+ \zeta_2+ \zeta_3 = \frac{1}{2}\sum_{k=1,2,3}(\nabla \varphi)(i_k^+) - (\nabla \varphi)(i_k^-) =  \sum_{k=1,2,3}(\nabla \varphi)(i_k^+) \in \ZZ.
	\] 
\end{proof}

\subsection{Dual descent algorithm: implementing SDA by submodular flow}\label{subsec:algo}

To check the existence of a circulation in ${\cal D}_{p,r}$, 
we construct an  instance of the maximum submodular flow problem.
Add a super source $a^+$ and super sink $a^-$.
For each edge $e = v^+ u^-$ in ${\cal D}_{p,r}$ having nonzero lower capacity $\underline{c}(e) > 0$,
replace $v^+u^-$ by two edges $v^+a^-$ and $a^+u^-$ 
with (upper) capacity $\underline{c}(e)$ (and lower capacity $0$).
Those edges are $i_0^+i_0^-$ for positive singular nodes~$i$ and
 $i_1^+i_2^-, i_2^+i_1^-$ for positive flat nodes $i$.
The resulting (skew-symmetric) network is denoted by $\tilde {\cal D}_{p,r}$, where
modified edge sets are denoted by $\tilde A_i$ $(i \in V)$ and the (upper) edge-capacity is denoted by $\tilde c$.
Consider the maximum $(a^+,a^-)$-submodular flow problem on $\tilde{\cal D}_{p,r}$, 
where submodular function $\rho$ on $U$ is given as
\begin{equation}\label{eqn:our_submo}
\rho(X) := \sum_{i:{\rm zero\,singular}} \mDelta^*_{i} (X \cap U_i) \quad (X \subseteq U).
\end{equation}
\begin{Lem}\label{lem:A}
    If $\{a^+\}$ is a minimum $(a^+,a^-)$-cut in $\tilde {\cal D}_{p,r}$, 
	then a circulation $\varphi$ in ${\cal D}_{p,r}$ exists, and is obtained from any maximum flow $\varphi'$ by the following procedure:
	\begin{itemize}
		\item[{\rm (A)}] For each edge $e = v^+u^-$ in ${\cal D}_{p,r}$ having nonzero lower capacity, 
		let $\varphi(e) := \varphi'(a^+u^-)$, and for other edge $e$ in ${\cal D}_{p,r}$, 
		let  $\varphi(e) := \varphi'(e)$.
	\end{itemize}
\end{Lem}
Indeed, in this case, any max-flow saturates all edges leaving $a^+$ and all edges entering $a^-$, 
and consequently the resulting $\varphi$ satisfies $\varphi(e) = \underline{c}(e) = \overline{c}(e)$ 
for all edges $e$ having nonzero lower-capacity, and is a circulation of ${\cal D}_{p,r}$.

Next we show that the minimal minimum cut gives rise to
a steepest descent direction of $g_{N, {\cal E}}$ at~$(p,r)$. 
An $(a^+,a^-)$-cut $X$ is said to be {\em normal} if it satisfies the following conditions:
\begin{itemize}
	\item[(c0)] $X$ does not meet $\{s^+,s^-\}$ for any terminal $s \in S$. 
	\item[(c1)] $X$ is a transversal.
	\item[(c2)] For each positive node $i$,  
	$X \cap U_i$ is equal to one of 
	\[
	\emptyset,\ U_i^+,\  U_i^-,\  \{i_k^+\},\ U_i^- \setminus \{i_k^-\},\ \{i_k^+\} \cup U_i^- \setminus \{i_k^-\} \ (k=1,2,3).
	\]
	\item[(c3)] For each zero node $i$,  $X \cap U_i$ is equal to one of
	\[
	\emptyset,\  U_i^+,\  \{i_k^+\},\  \{i_k^+\} \cup U_i^- \setminus \{i_k^-\}\ (k=1,2,3).
	\] 
\end{itemize}
Then the following holds; the proof is given in the end of this subsection.
\begin{Lem}\label{lem:normal_mincut}
	A normal minimum $(a^+,a^-)$-cut exists, and is obtained from the unique minimal minimum cut $X$
	by applying the following procedure:
	\begin{itemize}
		\item[{\rm (B)}] For each singular node $i$,
		if $|X \cap U^+_i| = 2$, then replace $X$ by  $X \cup U^+_i$. 
	\end{itemize}  
\end{Lem}

Let $U_{\cal I}$ be the subset of $U$ consisting of $i_k^+, i_k^-$
for $i \in V \setminus S$ and $k \in \{0, 1,2,3\}$ 
such that $(p(i),r(i)) \in W$, or $p(i) \in \mGamma^* \setminus \mGamma$ 
and $(p(i)_{\to^* k}, r(i)-1/2) \in W$.
Let $U_{\cal F}$ be the node subset defined by replacing $W$ with $B$
in the definition of $U_{\cal I}$.
A normal cut $X$ is said to be {\em ${\cal F}$-normal} if $X \cap U_{\cal I} = \emptyset$, 
and is said to be {\em ${\cal I}$-normal} if $X \cap U_{\cal F} = \emptyset$.
For an ${\cal F}$-normal or ${\cal I}$-normal cut~$X$, define $(p, r)^X$
by
\begin{equation}\label{eqn:X}
(p, r)^X(i) := \left\{
\begin{array}{ll}
(p(i)_{\to^*k}, r(i)+1/2) & {\rm if}\ X \cap U_i = \{i^+_k\},\\
(p(i)_{\to^*k}, r(i)-1/2) & {\rm if}\ X \cap U_i = U^-_i \setminus \{i^-_k\},\\
(p(i)_{\to k}, r(i)) & {\rm if}\ X \cap U_i = \{i_k^+\} \cup U^-_i \setminus \{i_k^-\},\\
(p(i), r(i) +1) & {\rm if}\ X \cap U_i = U_i^+, \\
(p(i), r(i) -1) & {\rm if}\ X \cap U_i = U_i^-, \\
(p(i), r(i)) & {\rm if}\ X \cap U_i = \emptyset
\end{array}\right.
\end{equation}
for each nonterminal node $i \in V \setminus S$ and $(p,r)^{X}(s) := (q_s,0)$ for each terminal $s \in S$.
\begin{figure} 
	\begin{center} 
		\includegraphics[scale=0.55]{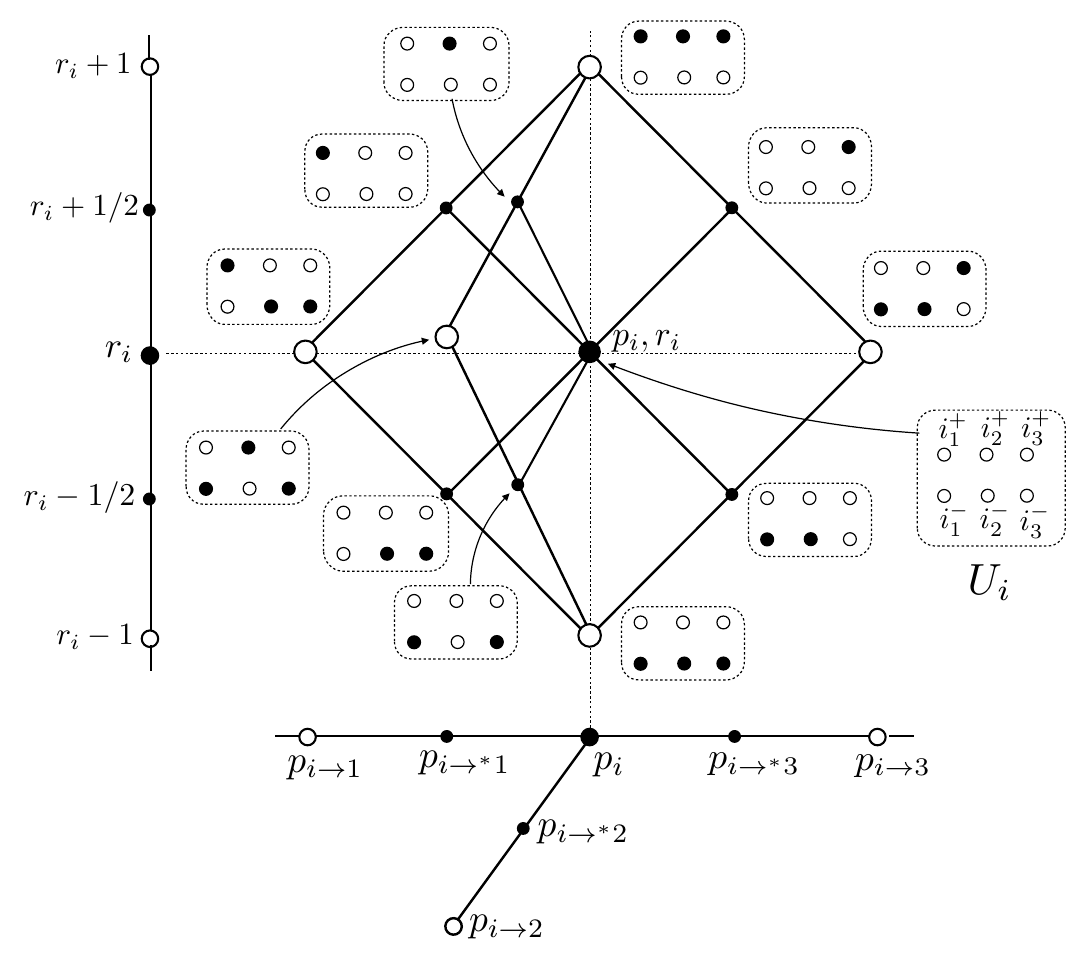}
		\caption{Neighborhood and normal cut}  
		\label{fig:neighborhood}         
	\end{center}
\end{figure}\noindent  
See Figure~\ref{fig:neighborhood} 
for the correspondence 
between $(p, r)^X(i)$ and $U_i \cap X$, 
where nodes in $U_i \cap X$ are colored black.
\begin{Prop}\label{prop:steepest}
	Let $X$ be a normal minimum $(a^+,a^-)$-cut in $\tilde {\cal D}_{p,r}$, and 
	let $X_{\cal F} := X \cap ( U_{\cal F} \cup \{a^+\})$ 
	and $X_{\cal I} := X \cap (U_{\cal I} \cup \{a^+\})$.  
	Then $(p, r)^{X_{\cal F}}$ is a minimizer of $g_{N, {\cal E}}$ over ${\cal F}_{p,r}$, and
	$(p, r)^{X_{\cal I}}$ is a minimizer of  $g_{N, {\cal E}}$ over ${\cal I}_{p,r}$.
\end{Prop}
The proof is given in the end.
We are now ready to describe our algorithm.
\begin{description}
	\item[Algorithm 3:] Dual descent algorithm (for a nondegenerate instance)
	\item[Input:] A nondegenerate instance $N = (V,E,S,c,a)$, ${\cal E} = (\mGamma, \{q_s\}_{s \in S})$.
	\item[Output:] A half-integral optimal multiflow $f$. 
	\item[Step 0:] Choose any vertex $v$ in $\mGamma_0$, and
	let $(p, r)$ be a potential defined by $(p(i),r(i)) := (v, d(\mGamma_0))$ for $i \in V \setminus S$ 
	and $(p(s),r(s)) := (q_s,0)$ for $s \in S$.
	\item[Step 1:] Construct network $\tilde {\cal D}_{p,r}$. Obtain an integral maximum $(a^+,a^-)$-flow $\varphi'$ 
	and the minimal minimum $(a^+,a^-)$-cut $X$ in $\tilde {\cal D}_{p,r}$ by a submodular flow algorithm.
	Make $X$ normal by Procedure (B).
	\item[Step 2:] If $X = \{a^+\}$, 
	then $(p,r)$ is optimal, obtain a feasible flow $\varphi$ in ${\cal D}_{p,r}$ from $\varphi'$
	by Procedure~(A).
	Obtain a $(p,r)$-admissible support $\zeta_{\varphi}$ by (\ref{eqn:zeta_phi}), 
	and obtain a half-integral optimal multiflow $f$ from $\zeta_{\varphi}$ by Algorithm 1; stop.
	\item[Step 3:] Choose $Y \in \{X_{\cal F}, X_{\cal I}\}$ with 
	$g_{N, {\cal E}}((p,r)^Y) = \min \{ g_{N, {\cal E}}((p,r)^{X_{\cal F}}), g_{N,{\cal E}}((p,r)^{X_{\cal I}}) \}$. Let~$(p,r) \leftarrow (p,r)^Y$, and go to step 1.
\end{description}

\begin{Thm}\label{thm:bound2}
The dual descent algorithm runs in $O( d(\mGamma_0) {\rm MSF} (n,m,1))$ time.
\end{Thm}
\begin{proof}
	By Proposition~\ref{prop:steepest}, 
	$(p,r)^{Y}$ is a steepest direction of $g_{N, {\cal E}}$ at $(p,r)$.
	Hence the dual descent algorithm is viewed as the steepest descent algorithm applied to $g_{N,{\cal E}}$ that is L-convex (Proposition~\ref{prop:g_is_L-convex}).
	By Lemma~\ref{lem:exists} and Theorem~\ref{thm:bound}, 
	the number of iterations is bounded by~$2 d(\mGamma_0)$.
	Our submodular function $\rho$ is a disjoint sum of submodular functions $\mDelta^*_{i}$ 
	for zero singular nodes $i$ (see (\ref{eqn:our_submo})).
	So the exchange capacity $\kappa(\cdot;u,v)$ for a pair of $u$ and $v$ 
	can take positive values if $u,v$ belong to $U_i$ for some zero singular node $i$, 
	and is equal to the exchange capacity for submodular function $\mDelta^*_{i}$ 
	on a 6-element set~$U_i$.
	 Hence this can be computed in constant time.
	 The number of nodes of $\tilde {\cal D}_{p,r}$ is at most $6n + 2$, 
	 and the number of edges is at most $2 m + 8 n$.
	 Thus step 2 is done in $O({\rm MSF}(n,m,1))$ time.
\end{proof}

\paragraph{Proof of the main result (Theorem~\ref{thm:main}).}

As in Section~\ref{sec:multiflow}, 
construct a nondegenerate instance $(\tilde N, \tilde {\cal E})$.
Then $d(\mGamma_0) \leq O(m \log k)$.
Apply the dual descent algorithm for  $(\tilde N, \tilde {\cal E})$.
By Theorem~\ref{thm:bound2}, 
we obtain an optimal potential $(\tilde p,\tilde r)$ 
and an optimal multiflow $f$ in $O((m \log k) {\rm MSF}(n,m,1))$ time.
As we have shown in Section~\ref{sec:multiflow}, $f$ is also a maximum multiflow,
and an optimal potential $(p,r)$ for the original instance is obtained from $(\tilde p, \tilde r)$ 
(in $O(n m \log k)$ time).
Then $r$ is a half-integral optimal solution of LP-dual (\ref{eqn:LP-dual}).
Indeed, $\sum_{i \in V \setminus S} c(i) r(i)$ is equal to the maximum flow-value.
The feasibility of $r$ follows from
$\sum_{i \in V(P) \setminus S} 2 r(i) = \sum_{ij \in E(P)} (r(i) + r(j)) 
\geq \sum_{ij \in E(P)} d(p_i,p_j) \geq 2$ for every $S$-path $P$.

\paragraph{Proof of Lemma~\ref{lem:normal_mincut}.}
	Let $X$ be the unique minimal minimum $(a^+,a^-)$-cut (of finite cut capacity).
	By Lemma~\ref{lem:minimal_minimum}, $X$ is a transversal, implying (c1).
	Also $X$ cannot meet $\{s^+,s^-\}$ for any terminal $s$.
	Otherwise, $X$ contains $s^-$ (and does not contain $s^+$). 
	However the deletion of $s^-$ from $X$ does not increase the cut capacity,
	contradicting the minimality.

	Consider a zero flat node $i$.
	Suppose that $i_k^- \in X$. 
	If $i_{k'}^+ \not \in X$ $(k' \neq k)$, then the deletion of $i_k^-$ from $X$ 
	does not increase the capacity, contradicting the minimality.
	Thus we have (c3).

	Consider a zero singular node $i$. Suppose that $X \cap U_i \neq \emptyset$.
	If $X \cap U_i^- = \{i_k^-\}$, 
	then $X \cap U_i$ is of type 1, 3, or 5, and 
	the change $X \to X \setminus \{i_k^-\}$
	preserves the type at $U_i$ and the cut capacity, contradicting the minimality.
	Similarly, if $X \cap U_i \subseteq U_i^-$, then the change $X \to X \setminus U_i^-$
	preserves the cut capacity.
	Thus $X \cap U_i \not \subseteq U_i^-$.
	So suppose that $U_i^- \cap X = \emptyset$ or $U_i^- \setminus \{i_k^-\}$.
	If $U_i^- \cap X  = U_i^- \setminus \{i_k^-\}$, then necessarily $U_i^+ = \{i_k^+\}$.
	Suppose that $U_i^- \cap X = \emptyset$.
	If $|X \cap U_i^+| \geq 2$ (type 1), then the change $X \to X \cup U_i^+$ (Procedure (B))
	does not increase the cut capacity (and keeps $X$ being a transversal). 
	In particular, the resulting $X \cap U_i$ is one of the patterns in (c3).

	Consider a positive singular node $i$.
	First we show that $|X \cap \{i_1^-, i^-_2, i_3^-\}| \neq 1$. 
	Suppose to the contrary that $X \cap \{i_1^-, i^-_2, i_3^-\} =  \{i_1^-\}$.
	The change $X \to X \setminus \{i_1^-, i_0^-\}$ preserves the capacity, contradicting the minimality.
    Also $X  \cap U_i^- = \{i_0^-\}$ is impossible, and 
    $|X \cap \{i_1^-, i^-_2, i_3^-\}| \geq 2$ implies $i_0^- \in X$.
	Thus the pattern of $X \cap U^-_i$ is one given in (c2).
	Next consider $X \cap U_i^+$. 
	If $X$ contains two nodes in $\{i^+_1,i^+_2, i^+_3\}$, and necessarily $X \cap U_i^- = \emptyset$; 
	the change $X \to X \cup U_i^+$ (Procedure (B))
	keeps $X$ being a transversal, and does not increase the cut capacity.
	If $X$ contains $i_0^+$, then $X$ contains at least two nodes in 
	$\{i_1^+,i_2^+,i_3^+\}$, reduced to the case above. 
	Thus, after Procedure~(B), the resulting cut is a normal minimum cut, as required.

\paragraph{Proof of Proposition~\ref{prop:steepest}.}
In $\mGamma^* \boxtimes \ZZ^*$,
vertices in $B$ are colored black, and vertices in $W$ are colored white.
Other vertices have no color.
For $(p,r) \in \mGamma^* \boxtimes \ZZ^*$
and $p' \in \mGamma^*$ with $p \neq p'$,
let $(p,r)_{\searrow p'} := (u, r - 1/2)$ for the unique neighbor $u$ of $p$ 
with $d(p,p') = d(u,p') + 1/2$. In addition, if $(p,r) \not \in B \cup W$, 
then let $(p,r)_{\swarrow p'} := (u, r - 1/2)$ for the (unique) neighbor $u$ of $p$ 
with $d(p,p') = d(u,p') - 1/2$.
\begin{Lem}\label{lem:parity}
	For $(p,r), (p',r') \in \mGamma^* \boxtimes \ZZ^*$ with $d(p,p') \geq 1$, 
	$d(p,p') - r -  r'$ is even if and only if  one of following pairs has the same color:
	\[
	((p,r), (p',r')), \ ((p,r)_{\searrow p'}, (p',r')),\  ((p,r), (p',r')_{\searrow p}),\ ((p,r)_{\searrow p'}, (p',r')_{\searrow p}).
	\]
\end{Lem}
\begin{proof}
	For $(p,r),(p',r') \in B \cup W$, observe that $d(p,p') - r - r' \equiv d(p,p_0) + r + d(p',p_0) + r' \mod 2$.
	Hence  $d(p,p') - r - r'$ 
	is even if and only if $(p,r)$ and $(p',r')$ 
	have the same color.
	Also observe that
   $d(p,p') - r -  r'$ does not change when  $(p,r)$ is replaced by $(p,r)_{\searrow p'}$. 
	The claim follows from these facts.
\end{proof}

\begin{Lem}\label{lem:noedge}
	There is no edge between $U_{\cal F}$ and $U_{\cal I}$.
	In particular, any normal cut $X$ is decomposed into 
	${\cal F}$-normal cut $X_{\cal F} := X \cap U_{\cal F}$ and 
	${\cal I}$-normal cut $X_{\cal I} := X \cap U_{\cal I}$ such that
	\[
	\tilde c(\delta X) - \tilde c(\delta \{a^+\}) = 
	 \tilde c(\delta X_{\cal F})  - \tilde c(\delta \{a^+\})
	+  \tilde c(\delta X_{\cal I})  - \tilde c(\delta \{a^+\}). 
	\]
\end{Lem}
\begin{proof}
	Edges of form $i^+_ki^-_{k'}$, $i^+_ki^+_{0}$, $i^-_0 i_k^-$ in $\tilde {\cal D}_{p,r}$
	belongs to $U_i$ for node $i$ of colored $(p_i, r_i)$, 
	and hence belongs to $U_{\cal F}$ or $U_{\cal I}$.
	So consider an edge $i_k^-j^+_{k'}$ for distinct $i, j$.
    Then $d(p_i,p_j) > 1$ (since $a_{ij} \geq 2$) and $d(p_i, p_j) - r_i -r_j$ is even (since $a_{ij}$ is even).
    By the previous lemma,  if both $(p_i,r_i)$ and $(p_j,r_j)$ are colored,
    then they have the same color.
    Suppose that $(p_i, r_i)$ has no color.
    If $(p_j, r_j)$ has a color, then $(p_i, r_i)_{\searrow p_j} = (p_{i \to^* k}, r_i -1/2)$ has the same color.
    If $(p_j, r_j)$ has no color, then 
    $(p_i, r_i)_{\searrow p_j} =  (p_{i \to^* k}, r_i -1/2)$ 
    and $(p_j, r_j)_{\searrow p_i} =  (p_{j \to^* k'}, r_j -1/2)$ have the same color.
     Consequently $i_k^-, j^+_{k'} \in U_{\cal F}$ or $i_k^-, j^+_{k'} \in U_{\cal I}$ for all cases.
\end{proof}

Let ${\cal F}^+_{p,r}$ (resp. ${\cal I}^+_{p,r}$) denote the subset of  
${\cal F}_{p,r}$ (resp. ${\cal I}_{p,r}$) consisting of $(p',r')$ with $r'_i \geq 0$ for $i=1,2,\ldots,n$.
Proposition~\ref{prop:steepest} follows from the above Lemma~\ref{lem:noedge} 
and the following.
\begin{Lem}
	\begin{itemize}
		\item[{\rm (1)}]	
		The map $X \to (p,r)^X$ is a bijection
		between ${\cal F}^+_{p,r}$ (resp. ${\cal I}^+_{p,r}$) 
		and the set of all ${\cal F}$-normal cuts  (resp. ${\cal I}$-normal cuts).
		\item[{\rm (2)}]
		For an ${\cal F}$-normal or ${\cal I}$-normal cut $X$, it holds
		\begin{equation}\label{eqn:g_NE(p,r)^X}
		g_{N,{\cal E}}((p,r)^X) - g_{N,{\cal E}}(p,r) = \tilde c( \delta X) + \rho(X \setminus \{a^+\})-   \tilde c(\delta \{a^+\}).
		\end{equation}
	\end{itemize}
\end{Lem}

\begin{proof}
	(1). We claim that ${\cal X}_i := \{X \cap U_i \mid \mbox{$X$ is an ${\cal F}$-normal cut} \}$ 
	and ${\cal F}_{p_i,r_i}^+$ are in one-to-one correspondence by~(\ref{eqn:X}).
	Suppose that $(p_i,r_i) \in W$, which implies $U_i \subseteq U_{\cal I}$.
	Then ${\cal F}^+_{p_i,r_i} = {\cal F}_{p_i,r_i} = \{(p_i,r_i)\}$ and 
	${\cal X}_i = \{ \emptyset \}$.
	Thus the claim is true.
	Suppose that $(p_i,r_i) \in B$, which implies $U_i \subseteq U_{\cal F}$.
	The claim can be seen from Figure~\ref{fig:neighborhood}.
	Suppose that $(p_i,r_i) \not \in B \cup W$; necessarily $i$ is positive.
	We can assume that $(p_{i \to^* 1}, r_i + 1/2) \in B$ (and $(p_{i \to^* 1}, r_i - 1/2) \in W$).
	Then ${\cal F}^+_{p_i,r_i} = \{ (p_i,r_i), (p_{i \to^* 1}, r_i -1/2), (p_{i \to^* 2}, r_i + 1/2) \}$.
	Since $i_1^+,i_1^- \in U_{\cal I}$ and $i_2^+,i_2^- \in U_{\cal F}$, 
	it holds that ${\cal X}_i = \{ \emptyset, \{i_2^-\}, \{i_2^+\}\}$.
	Thus we have the claim, implying the statement (1).

	(2). Let $(p',r') := (p,r)^X$.
	We first show that $g(p',r') < \infty$ if and only if $\tilde c(\delta X)$ is finite.
	Suppose that $g(p',r') < \infty$.
	Pick $i_k^- \in X$ and edge $i_k^-j_{k'}^+$ (of infinite capacity).
	We show $j_{k'}^+ \in X$.
	Recall that $d(p_i, p_j) - r_i - r_j - a_{ij} = 0$ and $a_{ij} \geq 2$.
	Hence $d(p_i,p_j) \geq 2$.
	Since $i_k^- \in X$, the change $(p_i,r_i) \to (p'_i,r'_i)$ increases
	$d(p_i, p_j) - r_i - r_j$ by one. 
	Necessarily the change $(p_j,r_j) \to (p'_j,r'_j)$ 
	 must decrease $d(p'_i, p_j) - r'_i - r_j$ by one.
	This means that 
	$(p'_j,r'_j) = (p_j, r_j+1)$, $(p_{j \to {k'}},r_j)$ or $(p_{j \to^*{k'}},r_j+1/2)$
	must hold.
	For all cases, $X$ contains~$j_{k'}^+$.

	Suppose that $\tilde c(\delta X)$ is finite.
	We show that $g(p',r') < \infty$. Here $r' \geq 0$ is clear.
	We need to show that $d(p'_i, p'_j)- r'_i - r'_j - a_{ij} \leq 0$ for each edge $ij \in E$.
	Suppose that $d(p_i, p_j)- r_i - r_j - a_{ij} = 0$. 
	Namely $ij \in E_{p,r}$ and $d(p_i, p_j) \geq 2$.
	In this case, the argument is the same as the above. Indeed,
	suppose that the change $(p_i,r_i) \rightarrow (p'_i,r'_i)$
	increases $d(p_i, p_j)- r_i - r_j$ by one.
	Then $X$ contains $i^-_k$ with $p_j \in \mGamma^*_{p_i, k}$, 
	and hence contains $j^+_{k'}$ with $p_i \in \mGamma^*_{p_j,k'}$.
	The change $(p_j,r_j) \rightarrow (p'_j,r'_j)$ decreases 
	$d(p'_i, p_j)- r'_i - r_j$ by one, implying $d(p'_i, p'_j)- r'_i - r'_j - a_{ij} = 0$.
	
	Thus it suffices to show that  
	$d(p'_i, p'_j)- r'_i - r'_j - a_{ij} = 1$ cannot occur.
	Otherwise, $d(p_i, p_j)- r_i - r_j - a_{ij} = - 1$, $d(p_i, p_j)- r_i - r_j$ is odd, 
	and $(p'_i, r'_i) \neq (p_i, r_i)$, $(p'_j, r'_j) \neq (p_j, r_j)$.
	If both $(p_i, r_i)$ and $(p_j, r_j)$ have colors, 
	then these colors are different (by Lemma~\ref{lem:parity}), 
   and $(p'_i, r'_i) = (p_i,r_i)$ or $(p'_j, r'_j) = (p_j,r_j)$ must hold; this is a contradiction.
   Suppose that $(p_i,r_i)$ has no color.
   Then $(p'_i, r'_i) = (p_i, r_i)_{\swarrow p_j}$ must hold.
   If $(p_j, r_j)$ has a color, then 
   this color must equal the color of $(p'_i, r'_i)$ 
   (since $(p_i, r_i)_{\swarrow p_j}$ and $(p_i, r_i)_{\searrow p_j}$ have different colors), 
   and hence we have a contradiction $(p'_j, r'_j) = (p_j, r_j)$.
   If $(p_j, r_j)$ has no color,  then $(p'_j, r'_j) = (p_j, r_j)_{\swarrow p_i}$ must hold; 
   but this is impossible since colors of  $(p_i, r_i)_{\swarrow p_j}$ and 
   $(p_j, r_j)_{\swarrow p_i}$ are different.

	Finally we show the equation~(\ref{eqn:g_NE(p,r)^X}).  The left hand side is equal to $\sum_{i \in V \setminus S} 2 c_i (r'_i - r_i)$ and the right hand side is equal to
	\begin{equation*}
	\sum_{i:{\rm positive}} \{  \tilde c( \tilde A_i \cap \delta X) - 2c(i)\} + \sum_{i: {\rm zero\, flat}} \tilde c( \tilde A_i \cap \delta X) + \sum_{i: {\rm zero\, singular}} \mDelta^*_{i}(X \cap U_i).
	\end{equation*}
    One can verify from the network construction (see Figure~\ref{fig:doublecovering}), 
    definitions of $(p,r)^X$ (see  (\ref{eqn:X})) and $\mDelta^*_c$ (see (\ref{eqn:delta*})) 
    that $2c_i (r'_i - r_i)$ is equal to $\tilde c(\tilde A_i \cap \delta X) - 2c(i)$ if $i$ is positive, 
    $\tilde c(\tilde A_i \cap \delta X)$ if  $i$ is zero flat, 
    and $\mDelta^*_{i}(X \cap U_i)$ if $i$ is zero singular.
	Thus we obtain the equation~(\ref{eqn:g_NE(p,r)^X}).
\end{proof}

\section*{Acknowledgments}
We thank Satoru Iwata  for information on submodular flow,  Satoru Fujishige for helpful comments, Yuni Iwamasa for computation of $\mDelta^*_b$.
and Motoki Ikeda for meticulously reading.
We also thank the referees for helpful comments.
The work was partially supported by JSPS KAKENHI 
Grant Numbers JP25280004, JP26330023, JP26280004, JP17K00029.

\end{document}